\documentclass[prodmode]{acmsmall-working}

\usepackage{amsmath,amsfonts,amssymb,bbm} 
\usepackage[numbers,sort&compress]{natbib} %
\usepackage{dsfont}
\usepackage{wrapfig}

\usepackage{subfig}
\usepackage[dvipsnames]{xcolor}
\usepackage{graphicx}
\usepackage{dcolumn}

\usepackage[colorlinks=true,breaklinks=true,bookmarks=true,urlcolor=MidnightBlue,citecolor=MidnightBlue,linkcolor=MidnightBlue,bookmarksopen=false,draft=false]{hyperref}

\newcommand{\tr}{\top}
\newcommand{\trace}{\mathrm{tr}}

\newcommand{\origmech}{\textsc{MSDG}}
\newcommand{\genmech}{\textsc{CA}}
\newcommand{\sign}{\operatorname{sign}}
\newcommand{\Sign}{\operatorname{Sign}}

\newcommand{\pr}{\mathrm{Pr}}

\newcommand{\DeltaSup}[1]{\Delta^{\hspace{-.1em} #1}} %

\newcommand{\trueS}{{S^\ast}}
\newcommand{\trueF}{\mathds{I}}
\newcommand{\trueG}{\mathds{I}}

\title{Informed Truthfulness in Multi-Task Peer Prediction}

\author{
VICTOR SHNAYDER
\affil{edX; Paulson School of Engineering, Harvard University} 
ARPIT AGARWAL
\affil{Indian Institute of Science, Bangalore}
RAFAEL FRONGILLO
\affil{University of Colorado, Boulder}
DAVID C. PARKES
\affil{Paulson School of Engineering, Harvard University}
}

\begin{document}

\begin{bottomstuff}
This research is supported in part by a grant from Google, the SEAS TomKat fund, and NSF grant CCF-1301976. Any opinions, findings, conclusions, or recommendations expressed here are those of the authors alone.
  Thanks to participants in seminars at IOMS NYU
  Stern, the Simons Instutite, the GSBE-ETBC seminar at Maastricht
  University, and reviewers
  for useful feedback.
Author addresses: 
\url{shnayder@eecs.harvard.edu},
\url{arpit.agarwal@csa.iisc.ernet.in},
\url{raf@colorado.edu},
\url{parkes@eecs.harvard.edu}.

This is the extended version of our EC'16 paper with the same title.
\end{bottomstuff}

\begin{abstract}
  The problem of peer prediction is to elicit information from agents
  in settings without any objective ground truth against which to
  score reports. Peer prediction mechanisms seek to exploit
  correlations between signals to align incentives with truthful
  reports. A long-standing concern has been the possibility of
  uninformative equilibria. For binary signals, a multi-task
  mechanism~\cite{DasguptaGhosh13} achieves {\em strong truthfulness},
  so that the truthful equilibrium strictly maximizes payoff. 
We  characterize conditions on the signal distribution
  for which this mechanism remains strongly-truthful with
  non-binary signals, also providing a greatly simplified proof. We  introduce the {\em
    Correlated Agreement} (CA) mechanism, which handles
 multiple
  signals and provides {\em informed truthfulness}: no strategy profile
  provides more payoff in equilibrium than truthful reporting, and
 the truthful equilibrium is strictly better than any uninformed strategy
  (where an agent avoids the effort of obtaining a signal). The
  $\genmech$ mechanism is maximally strongly truthful, in that no
  mechanism in a broad class of mechanisms is strongly truthful on
  a larger family of signal distributions.
    We also give a detail-free version of the  mechanism that 
removes any knowledge requirements on the part of the
designer, using reports on many tasks to learn
statistics  while retaining
$\epsilon$-informed truthfulness.
\end{abstract}
\maketitle

\section{Introduction}

We study the problem of information elicitation without verification
(``peer prediction''). This  challenging problem 
 arises across a diverse
range of multi-agent systems, in which participants are asked to
respond to an information task, and where there is no
external input available against which to score reports.  Examples
include completing surveys about the features of new products,
providing feedback on the quality of food or the ambience in a
restaurant, sharing emotions when watching video content, and peer
assessment of assignments in Massive Open Online Courses (MOOCs).

The challenge is to provide incentives for participants to
choose to invest effort in forming an opinion (a
``signal'') about a task, and to make truthful reports about
their signals. In the absence of inputs other than the reports of
participants, peer-prediction mechanisms make payments to one agent
based on the reports of others, and seek to align incentives by
leveraging correlation between reports (i.e., peers are rewarded for
making reports that are, in some sense, predictive of the reports of
others).

Some domains have binary signals, for example ``was a restaurant noisy
or not?'', and ``is an image violent or not?''. We are also interested in
domains with non-binary signals, for example:
\begin{itemize}
\item {\em Image labeling.} Signals could correspond to answers to
  questions such as ``Is the animal in the picture a dog, a cat or a
  beaver'', or ``Is the emotion expressed joyful, happy, sad or
  angry.'' These signals are categorical, potentially with some
  structure: `joyful' is closer to `happy' than `sad', for example.
\item {\em Counting objects.} There could be many possible signals, representing answers to questions such as (``are there 0, 1-5, 6-10, 11-100, or   $>$100 people in the picture''?). The signals
are ordered.
\item {\em Peer assessment in MOOCs}. Multiple students evaluate their peers' submissions to an open-response question using a  grading rubric. For example, an essay may be evaluated for clarity, reasoning, and relevance,  with the grade for reasoning ranging from 1 (``wild flights of fancy throughout''), through 3 (``each argument is well motivated and logically defended.'')
\end{itemize}

We do not mean to take an absolute position that external ``ground
truth'' inputs are never available in these applications. We do however believe it important to understand the extent to
which such systems can operate using only participant
reports. 

The design of peer-prediction mechanisms assumes the ability to make
payments to agents, and that an agent's utility is linear-increasing
with payment and does not depend on signal reports other than through
payment. Peer prediction precludes, for example, that an agent may
prefer to misreport the quality of a restaurant because she is
interested in driving more business to the restaurant.\footnote{The payments
need not be monetary; one could for example issue points to agents,
these points conveying some value (e.g., redeemable for awards, or
conveying status).  On a MOOC platform, the payments could correspond
to scores assigned as part of a student's overall grade in the class.
What is needed is a linear relationship between payment (of whatever
form) and utility, and expected-utility maximizers.}
The challenge of peer prediction is timely. For example, Google
launched {\em Google Local Guides} in November 2015. This provides
participants with points for contributing star ratings and
descriptions about locations. The current design rewards quantity but
not quality and it will be interesting to see whether this attracts
useful reports. After 200 contributions, participants receive a 1 TB upgrade of Drive storage (currently valued
at \$9.99/month.)

We are interested in {\em minimal} peer-prediction mechanisms, which require
only signal reports from participants.\footnote{While more complicated designs have
been proposed
(e.g.~\cite{Prelec2004,RBTS-Witkowski2012,radanovic-subjective-aaai15}),
in which participants are also asked to report their beliefs about the
signals that others will report, we believe that peer-prediction
mechanisms that require only signal reports are  more
likely to be adopted in practice. It is cumbersome to design user
interfaces for reporting beliefs, and people are notoriously bad at
reasoning about probabilities.}
A basic desirable property is
that truthful reporting of signals is a strict, correlated equilibrium
of the game  induced by the peer-prediction
mechanism.\footnote{It has been more common to refer to the
  equilibrium concept in peer-prediction as a Bayes-Nash equilibrium. But
as pointed out by Jens Witkowski, 
there is no agent-specific, private information
  about payoffs (utility is linear in payment).
In a correlated equilibrium, agents get signals and a strategy is a mapping from signals to actions. An action is a best response for a given signal if, conditioned on the signal, it maximizes an agent's expected utility. 
This equilibrium concept
fits peer prediction: each agent receives a signal from the environment, signals are correlated, and strategies map signals into reported signals.
}
For many years, an Achilles heel of peer prediction has been the
existence of additional equilibria that payoff-dominate truthful
behavior and reveal no useful
information~\cite{Jurca2009,DasguptaGhosh13,Radanovic-sensing2015}. An
uninformative equilibrium is one in which reports do not depend on the
signals received by agents.
Indeed, the equilibria of peer-prediction mechanisms must always
include an uninformative, mixed Nash equilibrium~\cite{waggoner14}.
Moreover, with binary signals, a single task, and two agents, \citeN{jurca-faltings2005} show that an
incentive-compatible, minimal peer-prediction mechanism will always
have an uninformative equilibrium with a higher payoff than truthful
reporting.  
Because of this, a valid concern has been that peer prediction could
have the unintended effect that agents who would otherwise be
truthful now adopt strategic misreporting behavior in order to
maximize their payments.

In this light, a result due to~\citeN{DasguptaGhosh13} is of interest: if agents are each
asked to respond to multiple, independent tasks (with some overlap
between assigned tasks), then in the case of binary signals there is a
mechanism that addresses the problem of multiple equilibria.  The
binary-signal, multi-task mechanism is {\em strongly truthful}, meaning that truthful
reporting yields a higher expected payment than any other strategy (and is
tied in payoff only with strategies that report permutations of 
signals, which in the binary case means $1\rightarrow 2, 2\rightarrow 1$).

We introduce a new, slightly weaker incentive property of {\em
  informed truthfulness}: no strategy profile provides more expected
payment than truthful reporting, and the truthful equilibrium is
strictly better than any uninformed strategy (where agent reports are
signal-independent, and avoid the effort of obtaining a signal).
Informed truthfulness is responsive to what we consider to be the two
main concerns of practical peer prediction design: \smallskip

(a) Agents should have strict incentives to exert effort toward acquiring
an informative signal, and

(b) Agents should have no incentive to misreport this information.
\smallskip

Relative to strong truthfulness, the relaxation to informed
truthfulness is that there may be other informed
strategies that match the expected payment of truthful reporting.
Even so, informed truthfulness retains the property of strong truthfulness that there can be no other behavior strictly better than truthful reporting.

The binary-signal, multi-task mechanism of Dasgupta and Ghosh is constructed from the simple building
block of a \emph{score matrix}, with a score of `1' for agreement and
`0' otherwise.
Some tasks are designated without knowledge of participants
as bonus tasks. The payment on a bonus task is 1 in the case of
agreement with another agent. There is
also a penalty of -1 if the agent's
report on another (non-bonus) task agrees with the report of another
agent on a third (non-bonus) task. In this way, 
the mechanism rewards agents when their reports on a shared (bonus) task agree more than would be expected based on their overall report frequencies.
  Dasgupta and
Ghosh %
remark that extending beyond two signals ``is one of the
most immediate and challenging directions for further work.''

Our main results are as follows:
\begin{itemize}
\item %
  We study the  {\em multi-signal extension of the Dasgupta-Ghosh mechanism} ($\origmech$), %
 and show that $\origmech$
is strongly truthful
for domains that are {\em categorical}, where
  receiving one signal reduces an agent's belief that other agents
  will receive any other signal. We also show that (i)  this
  categorical condition is tight for $\origmech$ for 
agent-symmetric signal
  distributions, and (ii) the peer grade distributions on a large
  MOOC platform do not satisfy the categorical property.
\item We generalize $\origmech$, %
 obtaining the {\em Correlated Agreement (CA) mechanism}. This 
 provides informed truthfulness in general domains, including
domains in which the $\origmech$ mechanism is neither 
informed- nor strongly-truthful.
The $\genmech$ mechanism requires  the designer to know the correlation structure of signals,
but not the full signal distribution.
  We further characterize domains where the $\genmech$ mechanism is
  strongly truthful, and show that no mechanism with similar structure and information requirements can do better.%
\item  For settings with a large number of tasks, we present a \emph{detail-free $\genmech$ mechanism}, in which the designer estimates the statistics of the 
correlation structure from agent reports. 
 This mechanism is
  informed truthful in the limit where the number of tasks is large
(handling the concern that reports affect estimation and thus scores),
and we provide
  a convergence rate analysis for $\epsilon$-informed truthfulness with
  high probability.
\end{itemize}

 We
believe that these are the first results on strong or informed
truthfulness 
in domains with
non-binary signals 
without
requiring a large population for their incentive properties (compare
with~\cite{Radanovic-sensing2015,Kamble2015,RFJ2016}).
The robust incentives of the multi-task $\origmech$ and $\genmech$ mechanisms hold for
as few as two agents and three tasks, whereas these previous papers
crucially rely on being able to learn statistics of the distribution
from multiple reports. Even if given the true underlying signal
distribution, the mechanisms in these
earlier papers would still need to use a large
population, with the payment rule based on statistics estimated from
reports, as this is critical for incentive alignment in these papers.
Our analysis framework also provides a dramatic simplification of the
techniques used by \citeN{DasguptaGhosh13}.

In a recent working paper,~\citeN{kong2016} show that a number of peer
prediction mechanisms that provide variations on strong-truthfulness
can be derived
within a single information-theoretic framework, with scores
determined based on the information they provide relative to reports
in the population (leveraging a measure of mutual information between
the joint distribution on signal reports and the product of marginal
distributions on signal reports).
Earlier mechanisms correspond to particular information measures.
Their results use different technical tools, and also include a
different, multi-signal generalization of \citeN{DasguptaGhosh13} that
is independent of our results,
outside of the family of mechanisms that we consider in Section~\ref{sec:52a}, and provides strong truthfulness in the
limit of a large number of tasks.\footnote{While they do not state or
  show that the mechanism does not need a large number of tasks in any
  special case, the techniques employed can also be used to design a
  mechanism that is a linear transform of our $\genmech$ mechanism, and
thus  informed truthful with a known signal correlation structure and a
  finite number of tasks (personal communication).}

\subsection{Related Work}

The theory of peer prediction has developed rapidly in recent years.
We focus on minimal peer-prediction mechanisms.  Beginning with
the seminal work of \citeN{MRZ2005}, a sequence of
results relax knowledge requirements on the part of the
designer~\cite{RBTS-Witkowski2012,jurca2011}, or generalize, e.g. to handle
continuous signal domains~\cite{radanovic-faltings14}.  Simple
output-agreement, where a positive payment is received if and only if
two agents make the same report (as used in the {\em ESP
game}~\cite{vonAhn2004}), has also received some theoretical
attention~\cite{waggoner14,Jain_acm_trans_econ_compu}.

Early peer prediction mechanisms had uninformative equilibria that gave better payoff than honesty. \citeN{Jurca2009} show how to remove
uninformative, pure-strategy Nash equilibria through a clever
three-peer design.
\citeN{ksl} show how to design strong truthful, minimal, single-task
mechanisms with a known model when there are reports from a large
number of agents.

In addition to \citeN{DasguptaGhosh13}
and~\citeN{kong2016}, 
several recent papers have tackled the problem of uninformative equilibria.
\citeN{Radanovic-sensing2015}
establish strong truthfulness amongst symmetric strategies in a
large-market limit where both the number of tasks and the number of
agents assigned to each task grow without bound. \citeN{RFJ2016}
provide complementary theoretical results, giving a mechanism in
which truthfulness is the equilibrium with highest payoff, based on
a population that is large enough to estimate
statistical properties of the report distribution. 
They require a self-predicting condition that limits the
correlation between differing signals.
Each agent need only be assigned
a single task.
\citeN{Kamble2015} describe a mechanism where truthfulness has higher payoff than uninformed strategies,  providing an asymptotic
analysis as the number of tasks grows without bound.
The use of learning is crucial in these papers.
In particular, they must use statistics estimated from reports to design
the payment rule in order to align incentives. This is a key distinction
from our work.\footnote{\citeN{cai-stat-estimate15} work in
a different model, showing how to 
achieve optimal statistical estimation from data
provided by self-interested participants. These authors
do not consider misreports and their mechanism is not informed- (or strongly-) truthful and is vulnerable to collusion. Their model is interesting, though,
in that it adopts a richer, non-binary effort model.
}
\citeN{Witkowski2013} first introduced the combination of learning and
peer prediction, coupling the estimation of the signal prior together
with the shadowing mechanism.

Although there is disagreement in the experimental literature about whether equilibrium selection is a problem in practice, there is compelling evidence that it matters~\cite{gao-ec14-trick-or-treat}; see~\citeN{faltings-hcomp14} for a study where uninformed equilibria did not appear to be a problem.\footnote{One difference is that this later study was in a many-signal domain, making it
harder for agents to coordinate on an uninformative strategy.}
\citeN{shnayder-ijcai16} use replicator dynamics as a model of agent learning to argue that equilibrium selection is indeed important, and that truthfulness is significantly more stable under mechanisms that ensure it has higher payoff than other strategies.
Orthogonal to concerns about equilibrium selection, \citeN{Gao2016}
point out a modeling limitation---when agents can coordinate on some
other, unintended source of signal, then this strategy may be better
than truthful reporting. They suggest randomly checking a fraction of
reports against ground truth as an alternative way to encourage
effort. We discuss this in Section~\ref{subsec:signal-models}.

Turning to online peer assessment for MOOCs, research has primarily
focused on evaluating students' skill at assessment and compensating
for grader bias~\cite{Piech2013}, as well as helping students
self-adjust for bias and provide better feedback~\cite{Kulkarni2013}.
Other studies, such as the {\em Mechanical TA}~\cite{Wright-KLB2015},
focus on reducing TA workload in high-stakes peer grading. A recent
paper~\cite{wu-las2015} outlines an approach to peer assessment that
relies on students flagging overly harsh feedback for instructor
review. We are not aware of any systematic studies of peer prediction
in the context of MOOCs, though \citeN{RFJ2016} present experimental
results from an on-campus experiment.

\section{Model}
\label{sec:model}

We consider two agents, 1 and 2, which are perhaps members of a
larger population. Let $k\in M=\{1,\ldots,m\}$ index a task from a
universe of $m\geq 3$ tasks to which one or both of these agents are
assigned, with both agents assigned to at least one task.  Each agent
receives a signal when investing effort on an assigned task. The
effort model that we adopt is binary: either an agent invests no effort and does not
receive an informed signal, or an agent invests effort
and incurs a cost and receives a signal.

Let $S_1,S_2$ denote random variables for the signals to agents 1 and
2 on some task. The signals have a finite domain, with $i,j\in
\{1,\ldots,n\}$ indexing a realized signal to agents 1 and 2,
respectively.

Each task is {\em ex ante} identical, meaning that pairs of signals
are i.i.d. for each task. Let $P(S_1{=}i,S_2{=}j)$ denote the
joint probability distribution on signals, with marginal probabilities
$P(S_1{=}i)$ and $P(S_2{=}j)$ on the signals of agents 1 and 2,
respectively. We assume exchangeability, so that the identity
of agents does not matter in defining the
signal distribution.
 The signal distribution is common knowledge to agents.\footnote{We assume common knowledge and symmetric signal models for simplicity of exposition. Our mechanisms do not require full information about the
signal distribution, only the correlation structure of signals, and can tolerate some user heterogeneity,
as described further in Section~\ref{sec:extensions}. %
}
We assume that the signal distribution satisfies \emph{stochastic relevance}, so
that for all $s'\neq s''$, there exists at least one signal $s$ such that 
\begin{align}
P(S_1{=}s|S_2{=}s')
\ne P(S_1{=}s|S_2{=}s''),
\end{align} 
and symmetrically, for agent 1's signal affecting the 
posterior on agent 2's. If two signals are
not stochastically relevant, they can be combined into
one signal.

Our constructions and analysis will make heavy use of the following matrix, which encodes the correlation structure of signals.
\begin{definition}[Delta matrix]
The {\em Delta matrix} $\Delta$ is an $n\times n$ matrix, 
with entry $(i,j)$ defined as
\begin{align}
\Delta_{ij} &= P(S_1{=}i, S_2{=}j) - P(S_1{=}i) P(S_2{=}j).
\end{align}
\end{definition}

The Delta matrix describes the correlation (positive or negative)
between different realized signal values.  For example, if
$\Delta_{1,2}=P(S_1{=}1,S_2{=}2)-P(S_1{=}1)P(S_2{=}2)=P(S_1{=}1)(P(S_2{=}2|S_1{=}1)-P(S_2{=}2))>0$, then $P(S_2{=}2|S_1{=}1)>P(S_2{=}2)$, so signal 2 is positively correlated with signal 1
(and by exchangeability, similarly for the effect of 1 on 2).
If a particular signal value increases the
probability that the other agent will receive the same signal then
$P(S_1{=}i,S_2{=}i)>P(S_1{=}i)P(S_2{=}i)$, and if this holds for all signals the Delta matrix has a positive diagonal.
Because the  entries in a row $i$ of
joint distribution $P(S_1{=}i,S_2{=}j)$ and 
a row of product distribution $P(S_1{=}i)P(S_2{=}j)$ both sum to
$P(S_1{=}i)$, each row in the $\Delta$ matrix sums to $0$ as the difference of the two.
 The same holds
for columns.

The $\genmech$ mechanism will depend on the sign structure of the $\Delta$ matrix, without knowledge of the specific values. We will use a sign operator $\Sign(x)$, with value 1 if $x>0$, 0 otherwise.\footnote{Note that this differs from the standard $\sign$ operator, which has value -1 for negative inputs.} 
\begin{example}
If the signal distribution is
\begin{align*}
P(S_1,S_2) &=
\begin{bmatrix}
     0.4       & 0.15 \\
     0.15       & .3
 \end{bmatrix}
\end{align*}

with marginal distribution $P(S) = [0.55; 0.45]$, we have
\[
\Delta = \begin{bmatrix}
    0.4       & 0.15 \\
    0.15       & .3
\end{bmatrix}
- 
\begin{bmatrix}
    0.55 \\
    0.45
\end{bmatrix}\cdot \begin{bmatrix}
0.55 & 0.45
\end{bmatrix}
\approx
\begin{bmatrix}
0.1 & -0.1 \\
-0.1  & 0.1
\end{bmatrix},
\text{ and }
\Sign(\Delta) = 
\begin{bmatrix}
1 & 0 \\
0  & 1
\end{bmatrix}.
\]
\end{example}

An agent's {\em strategy} defines, for every signal
it may receive and each task it is assigned, the signal it will
report. We allow for mixed strategies, so that an agent's strategy
defines a distribution over signals.  
Let $R_1$ and $R_2$ denote random variables for the {\em reports} by agents 1 and 2, respectively,
on some task.
Let matrices $F$ and $G$ denote the mixed strategies of agents 1 and 2, respectively,
with $F_{ir} = P(R_1{=}r|S_1{=}i)$ and $G_{jr} = P(R_2{=}r|S_2{=}j)$ to denote the
probability of making report $r$ given signal $i$ is observed
(signal $j$ for agent 2).
Let $r^k_1\in \{1,\ldots,n\}$ and $r^k_2\in \{1,\ldots,n\}$ 
refer to the realized report by agent 1 and 2, respectively,
on task $k$ (if assigned).
\begin{definition}[Permutation strategy]
A {\em permutation strategy}
is a deterministic strategy in which an agent adopts
a bijection between signals and reports, that is, $F$ (or $G$ for
agent 2) is a permutation matrix.
\end{definition}
\begin{definition}[Informed and uninformed strategies] 
An {\em informed strategy} has 
$F_{ir} \ne F_{jr}$ for some $i \ne j$, some $r\in \{1,\ldots,n\}$
(and similarly for $G$ for agent 2).
An {\em uninformed strategy} has the same report distribution for all signals.
\end{definition}
Permutation strategies are merely relabelings of the signals; in particular, truthfulness (denoted $\trueF$ below) is a permutation strategy.  Note also that by definition, deterministic uniformed strategies are those that give the same report for all signals.

Each agent is assigned to two or more tasks, and the agents overlap on
at least one task.  Let $M_b\subseteq M$ denote a non-empty set of
``bonus tasks'',  a subset of the tasks to which
both agents are assigned. Let $M_1\subseteq M\setminus
M_b$ and $M_2\subseteq M\setminus M_b$, with $M_1\cap M_2=\emptyset$ denote
non-empty sets of tasks to which agents 1 and 2 are assigned,
respectively. These will form the ``penalty tasks.''
For example, if both agents are assigned to each of
three tasks, $A, B$ and $C$, then we could choose $M_b=\{A\}$, 
$M_1=\{B\}$ and $M_2=\{C\}$. 

We assume that tasks are {\em a priori} identical, so that there is
nothing to distinguish two tasks other than their signals.  In
particular, agents have no information about which tasks are shared,
or which are designated bonus or penalty. This can be achieved by
choosing $M_b, M_1$ and $M_2$ randomly after task assignment. This can
also be motivated in largely anonymous settings, such as peer
assessment and crowdsourcing. 
A {\em multi-task peer-prediction mechanism} defines a total payment to each agent
based on the reports made across all tasks. The mechanisms
that we study assign
a total payment to an agent based on the sum of payments for each
bonus task, but where the payment for a bonus task is adjusted
downwards by the consideration of its report on a penalty
task and that of another agent on a different penalty task.

For the mechanisms we consider in this paper, it
is without loss of generality for each agent to adopt a uniform
strategy across each assigned task. Changing a strategy from task to
task is equivalent in terms of expected payment to adopting a linear
combination over these strategies, given that tasks are presented in a
random order, and given that tasks are equivalent, conditioned on
signal.

This result relies on the random order of tasks as presented to each agent, preventing coordination. Tasks will be indexed as $1, \ldots, k \ldots, m$ from the first agent's point of view. The second agent will see them reshuffled using a permutation $\pi$ chosen uniformly at random: $\pi(1), \ldots, \pi(m)$.

Let $\vec{F}$ be the first agent's strategy vector, with $F_k$ the first agent's strategy on task $k$. Fix the second agent's vector of strategies $\vec{G}$. Let $J_{ij}$ be the joint signal distribution. Then, for a broad class of mechanisms, it is without loss of generality to focus on agents having a single per-task strategy applied to all tasks. 

Let $K$, $K'$, $K''$ be random variables corresponding to a task id, with uniform probability of value $1,\ldots,m$. Let $\mathcal{M}$ be a \emph{linear} mechanism if its expected score function is a linear function of $\pr(R^K_1=r_1, R^K_2=r_2)$ and $\pr(R^{K'}_1=r_1,R^{K''}_2=r_2|K'\ne K'')$, for all set of report pairs $r_1,r_2$. For example, the DGMS mechanism we describe later has expected score 
\begin{align}
  \label{eq:2}
\pr(R^K_1&=R^K_2) - \pr(R^{K'}_1=R^{K''}_2|K'\ne K'')  = \\
&= \sum_{r=1}^n \pr(R^K_1=r,R^K_2=r) - \pr(R^{K'}_1=r,R^{K''}_2=r|K' \ne K''),
\end{align}

which fits this condition. The multi-task mechanism we define below is also linear. The expectation is with respect to the signal model, agent strategies, the random task order, and any randomization in the scoring mechanism itself.

\begin{lemma}
Let $\mathcal{M}$ be a linear mechanism.  Let $\vec{F}$ be a vector of strategies. Then for any $\vec{G}$, $\bar{F}=\text{mean}(\vec{F})$ will have the same expected score as $\vec{F}$.
\end{lemma}

\begin{proof}
We prove equivalence of expected value of $\pr(R^K_1=r_1, R^K_2=r_2)$ and $\pr(R^{K'}_1=r_1,R^{K''}_2=r_2|K' \ne K'')$ for all $r_1,r_2$, and equivalence for any $\mathcal{M}$ follows by linearity.

Fix $r_1,r_2$. We first show that $\pr(R^K_1=r_1, R^K_2=r_2)$ has the same expected value for $\vec{F}$ and $\bar{F}$. 

\begingroup
\allowdisplaybreaks
\begin{align}
\pr(R^K_1&=r_1, R^K_2=r_2) = \\
  &=\frac{1}{m} \sum_{k=1}^m \pr(R^k_1=r_1, R^k_2=r_2) \\
  &= \frac{1}{m} \sum_{k=1}^m \sum_{i=1}^n \sum_{j=1}^n \pr(S^k_1=i,S^k_2=j)\pr(R^k_1=r_1|s_1=i)\pr(R^k_2=r_2|s_2=j)\\
  &= \frac{1}{m} \sum_{k=1}^m\sum_{i=1}^n \sum_{j=1}^n J_{ij} F^k_{ir_1} G^{\pi(k)}_{jr_2},\\
  &\text{Taking the expectation over $\pi$, we get} \notag \\
  &= \frac{1}{m!} \sum_{\pi} \frac{1}{m}\sum_{k=1}^m \sum_{i=1}^n \sum_{j=1}^n J_{ij} F^k_{ir_1} G^{\pi(k)}_{jr_2} \\
\intertext{where the sum is over all $m!$ possible permutations of the tasks. 
By symmetry, we know that each element of $G$ will be used for task $k$ with equal probability $1/m$:} 
  &= \frac{1}{m} \sum_{\ell} \frac{1}{m}\sum_{k=1}^m \sum_{i=1}^n \sum_{j=1}^n J_{ij}  F^k_{ir_1} G^\ell_{jr_2} \label{key-line} \\
  & \text{and reordering the sums, we get:} \notag \\
  &= \frac{1}{m} \sum_{\ell}\sum_{i=1}^n \sum_{j=1}^n  J_{ij}  G^\ell_{jr_2} \frac{1}{m}\sum_{k=1}^m F^k_{ir_1}. \\
  &\text{Using the definition of $\bar{F}$ as the mean of $\vec{F}$,}\\
  &= \frac{1}{m} \sum_{\ell} \sum_{i=1}^n \sum_{j=1}^n  J_{ij}  G^\ell_{jr_2} \bar{F}_{ir_1} \\
  &= \pr(R^K_1=r_1,R^K_2=r_2|\text{using $\bar{F}$ instead of $\vec{F}$})
\end{align}%
\endgroup

The same argument works for $\pr(R^{K'}_1=r_1, R^{K''}_2=r_2|K' \ne K'')$, substituting $\pr(S_1=i)\pr(S_2=j)$ for $J_{ij}$. The key to the proof is the random permutation of task order in line~\ref{key-line}, which prevents coordination between the per-task strategies of the two agents.
\end{proof}

Given this uniformity, we write $E(F,G)$ to denote the expected payment to an agent for any 
bonus task. The expectation is taken with respect to both the
signal distribution and any randomization in agent
strategies.
Let $\trueF$ denote the truthful reporting strategy, which corresponds to the identity matrix.
\begin{definition}[Strictly Proper]
  A multi-task peer-prediction mechanism is %
 {\em proper} if and only if 
  truthful strategies
form a %
 correlated 
equilibrium,
so that
$E(\trueF,\trueG) \ge E(F,\trueG),$
for all strategies $F\neq \trueF$, and similarly 
when reversing the roles
of agents 1 and 2. For {\em strict properness}, the inequality must be strict.
\end{definition}
This insists that the expected payment on a bonus task is (strictly)
higher when reporting truthfully than when using any other
strategy, given that the other agent is truthful. 
\begin{definition}[Strongly-truthful] 
A multi-task peer-prediction mechanism is {\em strongly-truthful} if
and only if for all strategies $F, G$ we have
$E(\trueF, \trueG) \ge E(F,G),$
and equality may only occur when $F$
and $G$ are both the same permutation strategy.
\end{definition}
In words, strong-truthfulness requires that both agents being truthful has strictly greater expected payment than
any other strategy profile, %
unless both agents play the same
permutation strategy, in which case equality is allowed.\footnote{Permutation strategies seem unlikely to be a practical concern, since
permutation strategies require coordination and provide no benefit
over being truthful.}
From the definition, it follows that any strongly-truthful mechanism is strictly proper.
\begin{definition}[Informed-truthful] 
A multi-task peer-prediction mechanism is {\em informed-truthful} if
and only if for all strategies $F, G$,
$E(\trueF, \trueG) \ge E(F,G),$
and equality may only occur when
both $F$ and $G$ are informed strategies.
\end{definition}
In words, informed-truthfulness requires that the truthful strategy
profile has strictly higher expected payment than any profile in which
one or both agents play an uninformed strategy, and weakly greater
expected payment than all other strategy profiles. It follows that any
informed-truthful mechanism is proper.

Although weaker than strong-truthfulness, informed truthfulness is
responsive to the primary, practical concern in
peer-prediction applications: avoiding equilibria where agents achieve the same (or greater) payment as a truthful informed agent but without putting in the effort of forming a careful opinion about the task. For example, it would be undesirable
for agents to be able to do just as well or better by reporting the
same signal all the time.
Once agents exert effort and observe a signal, it is
reasonable to expect them to make truthful reports as long as this is
an equilibrium and there is no other equilibrium with higher expected payment. Informed-truthful peer-prediction mechanisms provide  this guarantee.\footnote{%
For simplicity of presentation, we do not model the cost of effort
explicitly, but it is a straightforward extension to handle the cost
of effort as suggested in previous work~\cite{DasguptaGhosh13}. In our proposed mechanisms, an
agent that does not exert effort receives an expected payment of zero, while the expected payment for agents that exert effort and
play the truthful equilibrium is strictly positive. With knowledge of
the maximum possible cost of effort, scaling the payments appropriately incentivizes effort.}

\section{Multi-Task Peer-Prediction Mechanisms}
\label{sec:multi-task-peer}

We define a class of multi-task peer-prediction mechanisms that is
parametrized by a {\em score matrix}, $S:
\{1,\ldots,n\}\times\{1,\ldots,n\}\to\mathbb{R}$, that maps a pair of
reports into a score, the same score for both agents.  This class of
mechanisms extends the binary-signal multi-task
mechanism due to~\citeN{DasguptaGhosh13} in a natural way.

\begin{definition}[Multi-task mechanisms]
These mechanisms are parametrized by score matrix $S$.
\begin{enumerate}
\item Assign each agent to two or more tasks, with at least one task in common, and at least three tasks total.
\item  Let $r^k_1$ denote the report received from agent 1
on task $k$ (and similarly for agent 2).
Designate one or more tasks assigned to both agents as bonus
tasks (set $M_b$). Partition the remaining tasks into penalty
tasks $M_1$ and
$M_2$, where $|M_1|>0$ and $|M_2|>0$ and $M_1$ tasks have a report from
agent 1 and $M_2$ a report from agent 2.
\item For each bonus task $k \in M_b$, pick a random $\ell \in M_1$ and
  $\ell' \in M_2$.  The payment to both agent 1 and agent 2 for task $k$
  is
$S(r^k_1,r^k_2)- S(r_1^{\ell}, r_2^{\ell'}).$
\item The total payment to an agent is the sum total payment across all bonus tasks.\footnote{A variation with the same expected payoff and the same
incentive analysis is to compute the expectation of the scores on all 
pairs of penalty tasks, rather than sampling. We adopt the
simpler design for ease of exposition. This alternate
design
would reduce score variance if there are many non-bonus tasks, and may be preferable in practice.}
\end{enumerate}
\end{definition}

As discussed above, it is important that agents do not know which
tasks will become bonus tasks and which become penalty tasks.
The expected payment on a bonus task
for strategies $F,G$ is
\allowdisplaybreaks
\begin{align}
E(F,G) &= \sum_{i=1}^n\sum_{j=1}^n  P(S_1{=}i, S_2{=}j) \sum_{r_1=1}^n \sum_{r_2=1}^n P(R_1{=}r_1|S_1{=}i)P(R_2{=}r_2|S_2{=}j) S(r_1,r_2) \notag\\
 &\hspace{-2em} - \sum_{i=1}^n\sum_{j=1}^n P(S_1{=}i) P(S_2{=}j)
 \sum_{r_1=1}^n \sum_{r_2=1}^n P(R_1{=}r_1|S_1{=}i)P(R_2{=}r_2|S_2{=}j)S(r_1,r_2)\notag\\
&=\sum_{i=1}^n\sum_{j=1}^n \Delta_{ij}
\sum_{r_1=1}^n\sum_{r_2=1}^n
S(r_1,r_2)F_{ir_1}G_{jr_2}.\label{eq:expected-score-delta-gen}
\end{align}

The expected payment can also be written succinctly as
$E(F,G) = \trace(F^\tr \Delta G S^\tr).$
In words, the expected payment on a bonus task is the sum, over all
pairs of possible signals, of the product of the correlation (negative
or positive) for the signal pair and the (expected) score given the
signal pair and agent strategies.  

For intuition, note that for the identity score matrix which pays \$1 in
the case of matching reports and \$0 otherwise, agents are incentivized to give matching reports for signal pairs with positive
correlation and non-matching reports for signals with negative
correlation.
Now consider a general score matrix $S$, and suppose that all agents always report 1.
They always get $S(1,1)$ and the expected value $E(F,G)$ is a multiple of the sum of entries in the
$\Delta$ matrix, which is exactly zero. Because individual rows and columns of $\Delta$ also sum to zero, this also holds whenever a single agent uses an uninformed strategy. In comparison, truthful behavior
provides payment $E(\trueF,\trueG)=\sum_{ij}\Delta_{ij}S(i,j)$, and will be positive if
the score matrix is bigger where signals are positively correlated than where they are not.

While agent strategies in our model can be randomized, the linearity
of the expected payments allows us to restrict our attention to
deterministic strategies.
\begin{lemma}
\label{lem:opt-is-deterministic}
For any world model and any score matrix $S$, there exists a deterministic, optimal joint strategy for a multi-task mechanism.
\end{lemma}
\begin{proof}
The proof relies on solutions to convex optimization problems being extremal. The game value can be written $V = \max_F \max_G h(F,G)$, where 
  \vspace{-0.3cm}
  \[h(F,G) = \sum_{i=1}^n\sum_{j=1}^n \Delta_{ij}
\sum_{r_1=1}^n\sum_{r_2=1}^n
S(r_1,r_2)F_{ir_1}G_{jr_2}~.\]
  Note that $h$ is linear in both $F$ and $G$ separately.  Now letting $V(F) = \max_G h(F,G)$ be the value for the $G$ player for a fixed $F$, we have $V = \max_F V(F)$ by definition.  As $h(F,\cdot)$ is linear, and the strategy space for $G$, all binary row-stochastic matrices, is convex, there exists a maximizer at an extreme point.  These extreme points are exactly the deterministic strategies, and thus for all $F$ there exists an optimal $G=G^{\text{opt}}$ which is deterministic.
 Considering the maximization over $F$, we see that $V(F) = \max_G h(F,G)$ is a pointwise supremum over a set of linear functions, and is thus convex.  $V$ is therefore optimized by an extreme point, some deterministic $F=F^{\text{opt}}$, and for that $F^{\text{opt}}$ there exists a corresponding deterministic $G^{\text{opt}}$ by the above.
\end{proof}

Lemma~\ref{lem:opt-is-deterministic} has several consequences:
\begin{itemize}
\item It is without loss of generality to focus on deterministic strategies when establishing strongly truthful or informed truthful properties of a mechanism.
\item There is a deterministic, perhaps asymmetric equilibrium, because the optimal solution that maximizes $E(F,G)$ is also an equilibrium.
\item It is without loss of generality to consider deterministic deviations when checking whether or not truthful play is an equilibrium.
\end{itemize}

We will henceforth assume deterministic strategies. By a slight abuse of notation, let $F_i\in \{1,\ldots,n\}$ and $G_j\in \{1,\ldots,n\}$ denote the reported signals by agent 1 for signal $i$ and agent 2 for signal $j$, respectively. The expected score then simplifies to
\begin{align}
E(F,G)&=\sum_{i=1}^n\sum_{j=1}^n \Delta_{ij} S(F_i,G_j).
\label{eq:expected-score-delta-determ}
\end{align}

We can think of deterministic strategies as mapping signal pairs to
reported signal pairs. Strategy profile $(F,G)$ picks out a report
pair (and thus score) for each signal pair $i, j$ with its
corresponding $\Delta_{ij}$. That is, strategies $F$ and $G$ map signals to
reports, and the score matrix $S$ maps reports to scores, so together
they map signals to scores, and we then dot those scores with
$\Delta$. 
\section{The Dasgupta-Ghosh Mechanism}

We first study the natural extension of the~\citeN{DasguptaGhosh13}
mechanism from binary to multi-signals.  This multi-task 
mechanism uses as the
score matrix $S$ the identity matrix (`1' for agreement, `0' for
disagreement.)
\begin{definition}[The Multi-Signal Dasgupta-Ghosh mechanism ($\origmech$)]
This is a multi-task mechanism with score matrix 
$S(i,j)=1$ if $i=j$, 0 otherwise.
\end{definition}

\begin{example}
Suppose agent 1 is assigned to tasks $\{A,B\}$ and
  agent 2 to tasks $\{B,C,D\}$, so that $M_b=\{B\}, M_1=\{A\}$ and
  $M_2=\{C,D\}$. Now, if the reports on $B$ are both 1, and the
  reports on $A, C$, and $D$ were $0,0$, and $1$, respectively, the
  expected payment to each agent for bonus 
task $B$ is $1 - (1\cdot 0.5 + 0 \cdot
  0.5) = 0.5$. In contrast, if both agents use an uninformed coordinating strategy and always report 1, the expected score for both is $1 - (1\cdot 0.5 + 1 \cdot 0.5)=0$. 
\end{example}

The expected payment in the $\origmech$ mechanism on a bonus task is
\begin{align}
E(F,G) = \sum_{i,j} \Delta_{ij} \mathds{1}_{[F_i=G_j]},\label{eq:expected-score-delta}
\end{align}
where $\mathds{1}_{x=y}$ is 1 if $x=y$, 0 otherwise.  An equivalent expression is $\trace(F^\tr \Delta G)$.

\begin{definition}[Categorical model]
A world model is {\em categorical} if, when an agent sees a signal, all other signals become less likely than their prior probability; i.e.,
$P(S_2=j|S_1=i) < P(S_2=j)$, for all $i$, for all $j\neq i$ (and analogously
for agent 2). This implies positive correlation for identical signals: $P(S_2=i|S_1=i) > P(S_2=i)$.
\end{definition}

Two equivalent definitions of categorical are that the Delta matrix has positive
diagonal and negative off-diagonal elements, or that $\Sign(\Delta)=\mathds{I}$.
\begin{theorem}
\label{thm:mech-is-strong truthful}
If the world is categorical, then the $\origmech$ mechanism is strongly truthful and strictly proper. Conversely, if the Delta matrix $\Delta$ is symmetric and the world is not categorical, then the $\origmech$ mechanism is not strongly truthful.
\end{theorem}
\begin{proof}
First, we show that truthfulness maximizes expected payment. We
have $E(F,G) = \sum_{i,j} \Delta_{ij} \mathds{1}_{[F_i=G_j]}$.
The truthful strategy corresponds to the identity matrix $\mathds{I}$,
and results in a payment equal to the trace of $\Delta$:
$E(\mathds{I},\mathds{I}) = \trace(\Delta) = \sum_{i} \Delta_{ii}$.
By the categorical assumption, $\Delta$ has positive diagonal and
negative off-diagonal elements, so this is the sum of all the positive
elements of $\Delta$. Because $\mathds{1}_{[F_i=G_j]} \le 1$, this is
the maximum possible payment for any pair of strategies.

To show strong truthfulness, first consider an asymmetric joint strategy, with $F \ne G$. Then there exists $i$ s.t. $F_i \ne G_i$, reducing the expected payment by at least $\Delta_{ii} > 0$. Now consider symmetric, non-permutation strategies $F=G$.  Then there exist $i \ne j$ with $F_i = F_j$. The expected payment will then include $\Delta_{ij} < 0$. This shows that truthfulness and symmetric permutation strategies are the only optimal strategy profiles. Strict properness follows from strong truthfulness.

For the tightness of the categorical assumption, first consider 
a symmetric $\Delta$ with positive off-diagonal elements $\Delta_{ij}$ and $\Delta_{ji}$. 
Then agents can benefit by both ``merging'' signals $i$ and $j$. Let $\bar{F}$ be the strategy that is truthful on all signals other than $j$, and reports $i$ when the signal is $j$. Then 
$ E(\bar{F},\bar{F}) = \Delta_{ij} + \Delta_{ji} + \trace(\Delta) > E(\mathds{I},\mathds{I}) = \trace(\Delta)$, so $\origmech$
is not strongly truthful.
Now consider a $\Delta$ where one of the on-diagonal entries is negative, say $\Delta_{ii}<0$. Then, because all rows and columns of $\Delta$ must add to 0, there must be a $j$ such that $\Delta_{ij} > 0$, and this reduces to the previous case where ``merging'' $i$ and $j$ is useful.
\end{proof}

For binary signals (`1' and `2'), any positively correlated model, such that
$\Delta_{1,1}>0$ and $\Delta_{2,2}>0$, is
categorical, and thus we obtain a substantially simpler proof of the main result in
Dasgupta and Ghosh~\shortcite{DasguptaGhosh13}.

\subsection{Discussion: Applicability of the $\origmech$ mechanism}

Which world models are categorical? One example is a noisy observation
model, where each agent observes the ``true'' signal $t$ with
probability $q$ greater than $1/n$, and otherwise makes a mistake
uniformly at random, receiving any signal $s \ne t$ with probability
$(1-q)/(n-1)$.
Such model makes sense for classification
tasks in which the classes are fairly distinct. For example, we would expect a categorical model
for a question such as ``Does the animal in this photo swim, fly, or walk?''

On the other hand, a classification problem such as the ImageNet
challenge~\cite{imagenet-ILSVRC15}, with 1000 nuanced and often
similar image labels, is unlikely to be categorical. For example, if
``Ape'' and ``Monkey'' are possible labels, one agent seeing ``Ape''
is likely to increase the probability that another says ``Monkey'',
when compared to the prior for ``Monkey'' in a generic set of photos.
The categorical property is also unlikely to hold when
signals have a natural order, which we dub \emph{ordinal} worlds.
\begin{example}
\label{example:non-categorical}
 If two evaluators grade essays on a scale from one to five, when one decides that an essay should get a particular grade, e.g. one,  this may increase the likelihood that their peer decides on that or an adjacent grade, e.g. one or two. In this case, the sign of the delta matrix would be
 \begin{equation}
   \label{eq:ordinal-delta}
\Sign(\Delta) = 
{\small
\begin{bmatrix}
1 & 1 & 0 & 0 & 0 \\
1 & 1 & 1 & 0 & 0 \\
0 & 1 & 1 & 1 & 0 \\
0 & 0 & 1 & 1 & 1 \\
0 & 0 & 0 & 1 & 1 \\
\end{bmatrix}.
}
 \end{equation}

Under the $\origmech$ mechanism, evaluators increase their expected payoff by agreeing to always report one whenever they thought the score was either one or two, and doing a similar ``merge'' for other pairs of reports. We will return to this example below.
\end{example}

The categorical condition is a stronger requirement than previously proposed properties in the literature, such as those assumed in the analyses of the \citeN{jurca2011} and \citeN{RFJ2016}
``1/prior'' mechanism and the
\citeN{RBTS-Witkowski2012} shadowing mechanism.
The 1/prior mechanism requires the self-predicting property
\[\pr(S_2=j|S_1=i) < \pr(S_2=j|S_1=j),\] 
whereas the categorical property
insists on a upper bound of $\pr(S_2=j)$, which is tighter than
$\pr(S_2=j|S_1=j)$ in the typical case where the model has positive correlation.%
The shadowing mechanism requires 
\[\pr(S_2=i|S_1=j) - \pr(S_2=i) <
\pr(S_2=j|S_1=j) - \pr(S_2=j),\]
 which says that the likelihood of signal
$S_2=i$ cannot go up ``too much'' given signal $S_1=j$, whereas the
categorical property requires the stronger condition that $\pr(S_2=i|S_1=j) - \pr(S_2=i) < 0$. 

To see how often categorical condition holds in practice, we look at the correlation structure in a dataset from a large
MOOC provider, focusing on 104 questions with over 100 submissions
each, for a total of 325,523 assessments from 17 courses.  Each
assessment consists of a numerical score, which we examine, and an
optional comment, which we do not study here.  As an example, one
assessment task for a writing assignment asks how well the student
presented their ideas, with options ``Not much of a style at all'',
``Communicative style'', and ``Strong, flowing writing style'', and a
paragraph of detailed explanation for each. These correspond to 0, 1,
and 2 points on this rubric element.\footnote{While we only see student reports, we take as an assumption
that these reasonably approximate the true world model.
As MOOCs develop along with valuable credentials based on their peer-assessed work, 
we believe it will nevertheless
become increasingly important to provide explicit credit mechanisms for 
peer assessment.}
\begin{figure}[t]
\begin{center}
    \includegraphics[width=0.4\textwidth]{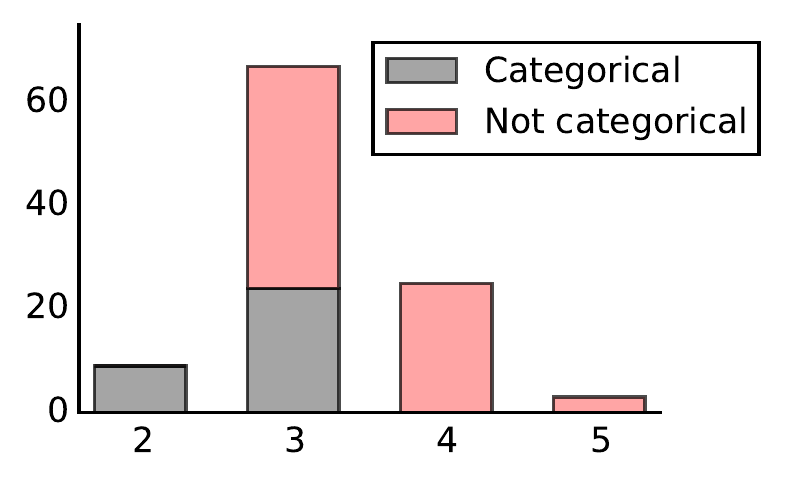}
\hfill
    \includegraphics[width=0.38\textwidth]{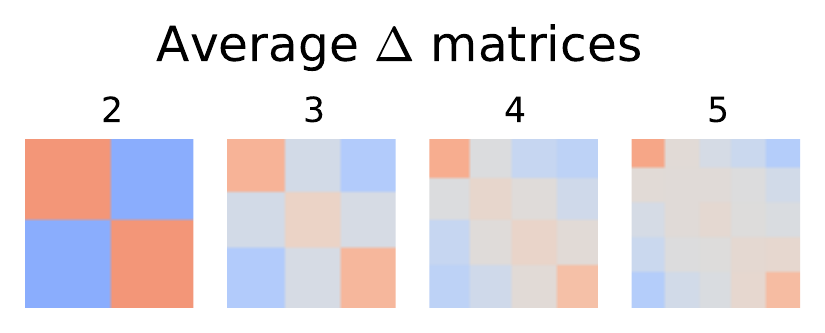}
    \includegraphics[width=0.1\textwidth]{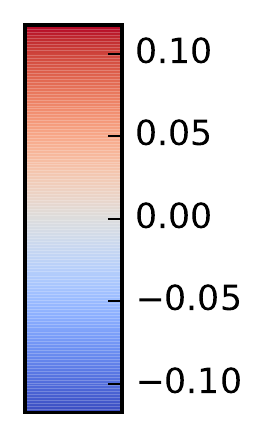}
\caption{
Left: MOOC peer assessment is an ordinal domain, with most models with three or more signals not categorical.
Right: Averaged $\Delta$ matrices, grouped by the number of signals in a domain. The positive diagonals show that users tend to agree on their assessments. For models of size 4 and 5, the ordinal nature of peer assessment is clear (e.g., an assessment of 2/5 is positively correlated with an assessment of 3/5).
}
\label{fig:model-breakdowns}
\vspace{-0.5cm}
\end{center}
\end{figure}

We estimate $\Delta$ matrices on each of the 104 questions from the
assessments. We can think about each question as corresponding to a
different signal distribution, and assessing a particular student's response to the question
as an information task that is performed by several peers.
The questions in our data set had five or fewer rubric
options (signals), with three being most common
(Figure~\ref{fig:model-breakdowns}L).

This analysis confirms that the categorical condition only holds for
about one third of our three-signal models and for none of the larger
models (Figure~\ref{fig:model-breakdowns}L). We also computed the
average $\Delta$ matrix for each model size, as visualized in
Figure~\ref{fig:model-breakdowns}R. The bands of positive correlation
around the diagonal are typical of what we refer to as an ordinal
rather than categorical domain.

\section{Handling the General Case}

In this section, we present a
mechanism that is informed-truthful for general domains.  We then discuss when it is strongly-truthful, give a version of it requiring no domain knowledge, and discuss other considerations.

\subsection{The Correlated Agreement Mechanism}

Based on the intuition given in Section~\ref{sec:multi-task-peer}, and the success of \origmech\ for categorical domains, it seems promising to base the construction of a mechanism on the correlation structure of the signals, and in particular, directly on $\Delta$ itself.  This is precisely our approach.  In fact, we will see that essentially the simplest possible mechanism following this prescription is informed-truthful for \emph{all} domains.
\begin{definition}[$\genmech$ mechanism]
The {\em Correlated Agreement ($\genmech$) mechanism} is a multi-task mechanism with score
matrix $S = \Sign(\Delta)$.
\end{definition}

\begin{theorem}
\label{thm:01-informed-truthful}
The $\genmech$ mechanism is informed-truthful and proper for all worlds.
\end{theorem}
\begin{proof}
The truthful strategy $F^\ast,G^\ast$ has higher payment than any other pair $F,G$:
\[E(F^\ast,G^\ast) = \sum_{i,j} \Delta_{i,j} S(i,j) = \sum_{i,j: \Delta_{ij}>0} \Delta_{i,j} \ge \sum_{i,j} \Delta_{i,j} S(F_i,G_j) = E(F,G),\]
where the inequality follows from the fact that $S(i,j) \in \{0,1\}$.

The truthful score is positive, while any uninformed strategy has score zero. Consider an uninformed strategy $F$, with $F_i=r$ for all $i$. Then, for any $G$,
\[E(F,G) = \sum_{i} \sum_j \Delta_{i,j} S(r,G_j) = \sum_j S(r,G_j) \sum_i \Delta_{i,j} = \sum_j S(r,G_j) \cdot 0 = 0, \]
where the next-to-last equality follows because rows and columns of $\Delta$ sum to zero.%
\end{proof}

While informed-truthful, the \genmech\ mechanism is not always strictly proper.  As discussed at the end of Section~\ref{sec:model}, we do not find this problematic; let us revisit this point.
The peer prediction literature makes a distinction between proper and strictly proper, and insists on
the latter. This comes from two motivations: (i) properness
is trivial in standard models: one can simply pay the same amount all
the time and this would be proper (since truthful reporting would be
as good as anything else); and (ii) strict properness provides
incentives to bother to acquire a useful signal or belief before
making a report. Neither (i) nor (ii) is a critique of 
 the $\genmech$ mechanism; consider (i) paying a fixed
amount does not give informed truthfulness, and (ii) the
mechanism
provides strict incentives to invest effort in acquiring a signal.
\begin{example}
Continuing with Example~\ref{example:non-categorical}, we can see why $\genmech$ is not manipulable.
$\genmech$ considers signals that are positively
correlated on bonus tasks (and thus have a positive entry in
$\Delta$) to be matching, so there is no need to agents to misreport
to ensure matching. 
In simple cases, e.g. if only the two signals 1 and 2 are positively correlated, they are ``merged,'' and reports of one treated equivalently to the other. In cases such as Equation~\ref{eq:ordinal-delta}, the correlation structure is more complex, and the result is not simply merging. 
\end{example}

\subsection{Strong Truthfulness of the $\genmech$ Mechanism}
\label{sec:52a}

The $\genmech$ mechanism is always informed truthful.  In this section we
characterize when it is also strongly truthful (and thus strictly proper), and show that it is 
maximal in this sense across a large class of mechanisms.
\begin{definition}[Clustered signals]
\label{def:clus-sig}
A signal distribution has {\em clustered signals} when there exist at least two identical rows or columns in $\Sign(\Delta)$.
\end{definition}
  
\begin{figure}[t]
\begin{center}
    \includegraphics[width=0.4\textwidth]{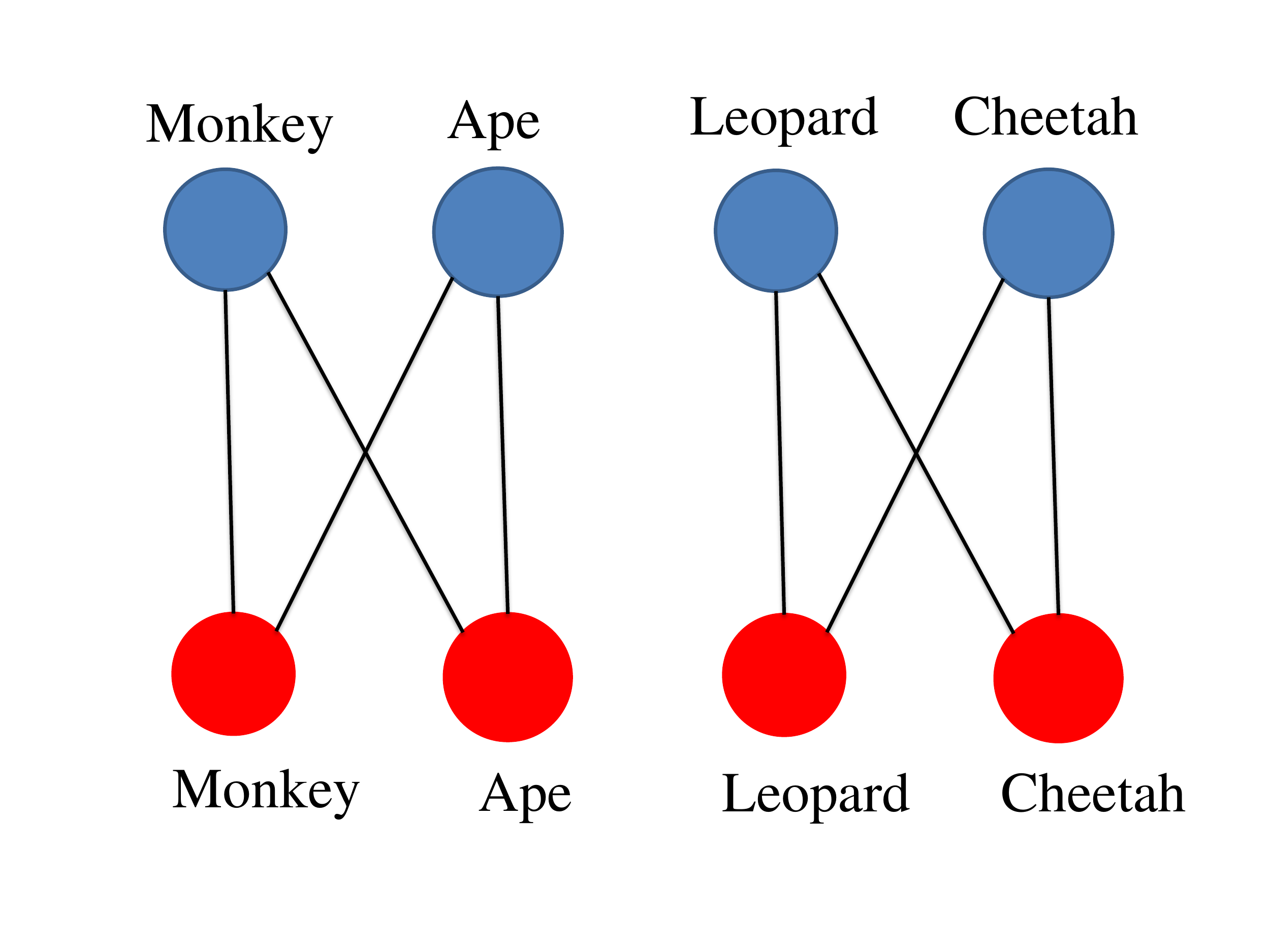}
\hfill
    \includegraphics[width=0.4\textwidth]{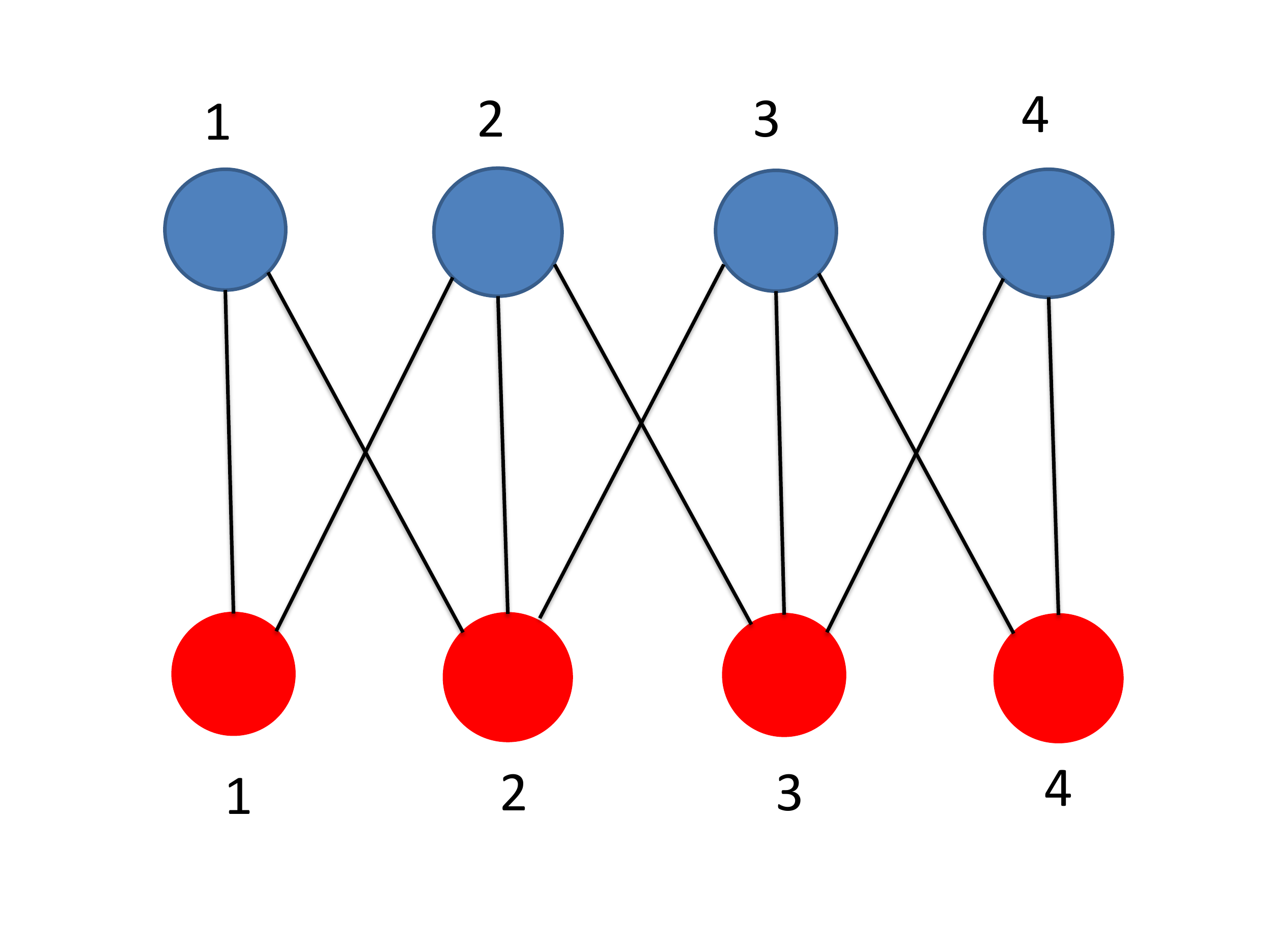}
\caption{
The blue and red nodes represent signals of agent 1 and 2, respectively. An edge between two signals represents that there is positive correlation between those signals.  
Left: A signal distribution for an image classification task with clustered signals.
Right: A signal distribution for a MOOC peer assessment task or object counting task with ordinal signals and without clustered signals.
}
\label{fig:clus-sig}
\end{center}
\end{figure}

Equivalently, two signals $i$ and $i'$ of an agent are 
clustered if $i$ is positively correlated with the same 
set of matched agent's signals as $i'$.
\begin{example}
See
Figure~\ref{fig:clus-sig}.
The first example corresponds to an image classification task where there are categories such as ``Monkey'', ``Ape'', ``Leopard'', ``Cheetah'' etc. 
The signals ``Monkey'' and ``Ape'' are clustered: for each agent, seeing one is positively correlated with the other agent having one of the two, and negatively correlated with the other possible signals. 
The second example concerns models with ordinal signals, such as 
peer assessment or counting objects.
In this  example there are no clustered signals for either agent. For example,  signal 1 is positively correlated with signals 1 and 2, while signal 2 with signals 1, 2, and 3.
\end{example}
\begin{lemma}
\label{lem:clus-sig}
  If $\Delta_{ij} \neq 0$, $\forall i,j$, then a joint strategy where at least one agent uses a non-permutation strategy and matches the expected score of truthful reporting exists if and only if there are clustered signals.
\end{lemma}
\begin{proof}
Suppose clustered signals, so there exists $i\neq i'$ such that 
$\Sign(\Delta_{i,\cdot}) = \Sign(\Delta_{i',\cdot})$. Then if agent 2 is truthful, agent 1's expected score is the same for being truthful or for reporting $i'$ 
whenever she receives either $i$ or $i'$.
Formally, consider the strategies $G = \trueG$ and $F$ formed by replacing the $i$-th row in $\trueF$ by the $i'$-th row.  Observe that $S(i, j) = S(F_i, G_j)$ as the $i$-th and $i'$-th row in $S$ are identical. Hence, $E(F,G) = E(\trueF, \trueG)$. 
The same argument holds for clustered signals for agent 2.

If the world does not have clustered signals, any agent using a non-permutation strategy leads to lower expected score than being truthful.
Suppose $F$ is a non-permutation strategy, such that $E(F,G) = E(\trueF, \trueG)$ for some $G$. Then there exist signals $i\neq i'$ such $F_i=F_{i'}=r$, for some $r$. No clustered signals implies that $\exists j$ such that $\Sign(\Delta_{i,j}) \neq \Sign(\Delta_{i',j})$. Let $G(j) = j'$, for some $j'$. Without loss of generality assume that $\Delta(i,j) > 0$, then we get $\Delta(i',j) < 0$ as $\Delta(i',j) \neq 0$.
The score for signal pair $(S_1 = i,S_2 = j)$ is $S(r,j')$ and for $(S_1 = i',S_2 = j)$ is also $S(r,j')$. Either $S(r,j') = 1$ or $S(r,j') = 0$. In both cases the strategy profile $F,G$ will lead to a strictly smaller expected score as compared to the score of truthful strategy, since $\Delta(i,j) > 0$ and $\Delta(i',j) < 0$.  
Similarly, we can show that if the second agent uses a non-permutation strategy, that also leads to strictly lower expected scores for both agents.
\end{proof}

We now give a condition under which there are asymmetric permutation strategy profiles that give the same expected score as truthful reporting.
\begin{definition}[Paired permutations]
  A signal distribution has 
{\em paired permutations}
   if there exist distinct permutation matrices $P, Q$ 
   \,s.t. $P \cdot \Sign(\Delta) = \Sign(\Delta) \cdot Q$.  
\end{definition}
\begin{lemma}
\label{lem:paired-perm}
If $\Delta_{ij} \neq 0$, $\forall i,j$, then there exist asymmetric permutation strategy profiles with the same expected score under the $\genmech$ mechanism as truthful reporting if and only if the signal distribution has paired permutations.
\end{lemma}

\begin{proof}
First we show that if the world has paired permutations then there exist asymmetric permutation strategy profiles that have the same expected score as truthful strategies.
Consider $F = P$ and $G = Q$. From the paired permutations condition it follows that $S(i,j) = S(F_i, G_j)$, $\forall i,j$, since $S(F_i, G_j)$ is the $(i,j)$-th entry of the matrix $F\cdot S \cdot G^\top$ which is equal to $S$. Therefore, $E[F,G] = E[\trueF, \trueG]$.

To prove the other direction, let $F$ and $G$ be the permutation strategies of agent 1 and 2, respectively, with $F  \neq G$. If the world does not have paired permutations, then $F \cdot S \cdot G^\top \neq S$. Let $\hat{S} = F \cdot S \cdot G^\top$. The expected score for $F,G$ is 
\[
 E[F,G] = \sum_{i,j} \Delta_{i,j} \cdot \hat{S}(i,j)
 \,,
\]
and the expected score for truthful strategies is 
\[
 E[\trueF,\trueG] = \sum_{i,j} \Delta_{i,j} \cdot S(i,j)
 \,.
\]
Combining the facts that $E[\trueF, \trueG] \geq E[F,G]$; $\Delta_{ij} \neq 0$, $\forall i,j$; and $\hat{S}$ differs from $S$ by at least one entry, $E[F,G]$ will be strictly less than $E[\trueF,\trueG]$. 
\end{proof}

Lemma~\ref{lem:clus-sig} shows that when the world has clustered signals, the $\genmech$
mechanism cannot differentiate between individual signals in a
cluster, and is not strongly truthful. Similarly, Lemma~\ref{lem:paired-perm} shows that under paired permutations this mechanism is not able to distinguish whether an agent is reporting the true signals or a particular permutation of the signals. In domains without clustered signals and paired permutations, all strategies (except symmetric permutations) lead to a strictly lesser score than truthful strategies, and hence, the $\genmech$ mechanism is strongly truthful. 

The $\genmech$ mechanism is informed truthful, but not strongly truthful, for the image classification example in Figure~\ref{fig:clus-sig} as there are clustered signals in the model. For the peer assessment example, it is strongly truthful because there are no clustered signals and a further analysis reveals that there are no paired permutations.

A natural question is whether we can do better by somehow
`separating' clustered signals from each other, and `distinguishing' permuted signals from true signals, by giving different scores to
different signal pairs, while retaining the property that
the designer only needs to know $\Sign(\Delta)$.  Specifically,
can we do better if we allow the score for each signal pair $(S_1 = i, S_2 = j)$ to
depend on $i,j$ in addition to $\Sign(\Delta_{ij})$?
We show that this extension does not add any additional power over the $\genmech$ mechanism in terms of strong truthfulness. 

\begin{theorem}
\label{thm:maximally-strong-truthful}
 If $\Delta_{ij} \neq 0$, $\forall i,j$, then $\genmech$ is maximally strong truthful amongst multi-task mechanisms that
only use knowledge of the correlation structure of signals, i.e.\ mechanisms that decide $S(i,j)$ using $\Sign(\Delta_{ij})$ and index $(i,j)$.
\end{theorem}

\begin{proof}

We first show that the $\genmech$ mechanism is strongly truthful if the signal distribution has neither clustered signals nor paired permutations. This follows directly from Lemmas~\ref{lem:clus-sig} and~\ref{lem:paired-perm}, as strategy profiles in which any agent uses a non-permutation strategy or both agents use an asymmetric permutation strategy lead to strictly lower expected score than truthful strategies.

Next we show maximality by proving that 
if a signal distribution has either clustered signals or paired permutations then 
there do not exist any strong truthful multi-task mechanisms that only use the correlation structure of signals.

We prove this by contradiction. Suppose there exists a strongly truthful mechanism for the 
given signal distribution which computes the scoring matrix using the correlation structure of signals. Let the scoring matrix for the signal distribution be $S$.

If the signal distribution has clustered signals then at least two rows or columns in $\Sign(\Delta)$ are identical. 

Suppose that there exist $i \neq i'$, such that the $i$-th and $i'$-th row in $\Sign(\Delta)$ are identical. 
We will construct another delta matrix $\Delta'$ representing a signal distribution that has clustered signals, for which this mechanism cannot be simultaneously strongly truthful. 

Let $\Delta'$ be computed by exchanging rows $i$ and $i'$ of $\Delta$. 
Clearly, $\Delta'$ has clustered signals. Now, the scoring matrix for both $\Delta$ and $\Delta'$ is the same, since the sign structure is the same for both. Let $G = \trueG$ and $F$ be computed by exchanging rows $i$ and $i'$ of $\trueF$.

Strong truthfulness for $\Delta$ implies that

\begin{equation}
 \label{eqn:strong-delta}
 E_{\Delta}[\trueF,\trueG] > E_{\Delta}[F,G]
 \,.
\end{equation}

However, observe that $E_{\Delta}[\trueF,\trueG] = E_{\Delta'}[F,G]$ 
and $E_{\Delta'}[\trueF,\trueG] = E_{\Delta}[F,G]$. 
Strong truthfulness for $\Delta'$ implies that

\begin{equation}
 \label{eqn:strong-delta'}
 E_{\Delta'}[\trueF,\trueG] > E_{\Delta'}[F,G] \implies E_{\Delta}[\trueF,\trueG] < E_{\Delta}[F,G]
 \,.
\end{equation}

Equation~\ref{eqn:strong-delta} and~\ref{eqn:strong-delta'} lead to a contradiction, implying that the above mechanism cannot be strongly truthful.

Similarly, we can show that if two columns in $\Sign(\Delta)$ are identical,
then there exists another 
delta matrix $\Delta'$ formed by exchanging the columns of
the $\Delta$ for $j\neq j'$ such that the $j$-th and $j'$-th column of $\Sign(\Delta)$ are identical. 
A similar contradiction can be reached using strong truthfulness on $\Delta$ and $\Delta'$.

The interesting case is when the signal distribution satisfies paired permutations,
i.e.\ there exist permutation matrices $P\neq Q$ such that $P\cdot S \cdot Q^\top = S$. Consider $\Delta' = (P^{-1}) \cdot \Delta \cdot (Q^{-1})^\top$, $F = P$, and $G = Q$. We need to argue that $\Delta'$ represents a correct signal distribution and that it has paired permutations.

To see this, observe that exchanging the columns or rows 
of a delta matrix leads to a valid delta matrix, 
and pre-multiplying or post-multiplying a matrix with 
permutation matrices only exchanges rows or columns, respectively. 
Observe that the sign structure of $\Delta'$ is 
the same as the sign structure of $\Delta$ since 
$S = (P^{-1}) \cdot S \cdot (Q^{-1})^\top$, and therefore, 
the scoring matrix for both $\Delta$ and $\Delta'$ is the same. 
Due to this $\Delta'$ has paired permutations.

Strong truthfulness for $\Delta$ implies that

\begin{equation}
 \label{eqn:strong-delta-2}
 E_{\Delta}[\trueF,\trueG] > E_{\Delta}[F,G]
 \,.
\end{equation}

However, again observe that $E_{\Delta}[\trueF,\trueG] = E_{\Delta'}[F,G]$ 
and $E_{\Delta'}[\trueF,\trueG] = E_{\Delta}[F,G]$ 
Strong truthfulness for $\Delta'$ implies that

\begin{equation}
 \label{eqn:strong-delta'-2}
 E_{\Delta'}[\trueF,\trueG] > E_{\Delta'}[F,G] \implies E_{\Delta}[\trueF,\trueG] < E_{\Delta}[F,G]
 \,.
\end{equation}

Equation~\ref{eqn:strong-delta-2} and~\ref{eqn:strong-delta'-2} lead to a contradiction, implying that the above mechanism cannot be strongly truthful.

Therefore, if the signal distribution has either clustered signals or paired permutations there exist no strongly truthful scoring mechanism that assigns scores based on the correlation structure of $\Delta$.
\end{proof}

This result shows that if a multi-task mechanism only relies on the correlation structure and is strongly truthful in some world model 
then the $\genmech$ mechanism will also be strongly truthful in that world model. 
Therefore, even if one uses $2\cdot n^2$ parameters in the design of scoring matrices from $\Sign(\Delta)$, one can only be strongly truthful in the worlds where $\genmech$ mechanism is strongly truthful, which only uses 2 parameters. 

\begin{figure}
\centering
    \includegraphics[width=0.4\textwidth]{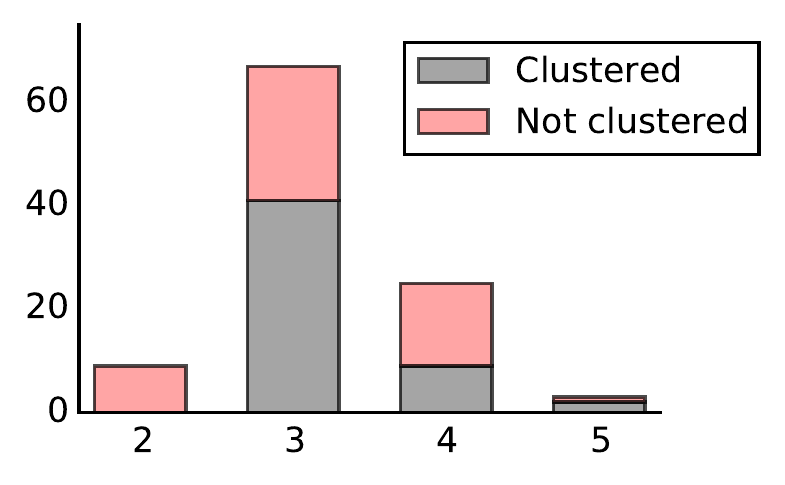}
\caption{Number of MOOC peer assessment models with clustered signals ($\genmech$ is informed truthful) and without clustered signals ($\genmech$ is strongly truthful up to paired permutations).
\label{fig:informed-strong-breakdown}}
\end{figure}

A remaining question is whether strongly truthful mechanisms 
can be designed when the score matrix can depend on the exact value of the $\Delta$ matrix.  We answer this question negatively.
\begin{theorem}
\label{lem:strong-truthful-impossibility}
There exist symmetric signal distributions such that no multi-task mechanism is strongly truthful.
\end{theorem}

\begin{proof}
Let $n=3$, and consider any symmetric $\Delta$ matrix of the form:

\[
\Delta = 
\begin{bmatrix}
x & y & -(x+y) \\
y & x & -(x+y) \\
-(x+y) & -(x+y) & 2(x+y) \end{bmatrix}
  \,,
\]

for some $0 < y<x \leq 0.5$, and let 

\[
S =
\begin{bmatrix}
a & b & e \\
c & d & f \\
g & h & i \end{bmatrix}
  \,,
\]

for some $a,b,c,d,e,f,g,h,i$ which can be selected using complete knowledge of $\Delta$.

We will consider three strategy profiles $(F^1,G^1), (F^2,G^2), (F^3,G^3)$, with

\[
F^1 = 
\begin{bmatrix}
0 & 1 & 0  \\
1 & 0 & 0 \\
0 & 0 & 1 \end{bmatrix}
\qquad
\qquad
G^1 = 
\begin{bmatrix}
1 & 0 & 0  \\
0 & 1 & 0 \\
0 & 0 & 1 \end{bmatrix}
\;,
\]

\[
F^2 = 
\begin{bmatrix}
1 & 0 & 0  \\
1 & 0 & 0 \\
0 & 0 & 1 \end{bmatrix}
\qquad
\qquad
G^2 = 
\begin{bmatrix}
1 & 0 & 0  \\
1 & 0 & 0 \\
0 & 0 & 1 \end{bmatrix}
  \,,
\]

and

\[
F^3 = 
\begin{bmatrix}
0 & 1 & 0  \\
0 & 1 & 0 \\
0 & 0 & 1 \end{bmatrix}
\qquad
\qquad
G^3 = 
\begin{bmatrix}
0 & 1 & 0  \\
0 & 1 & 0 \\
0 & 0 & 1 \end{bmatrix}
  \,.
\]

Using strong truthfulness condition $E[\trueF, \trueG] > E[F^1, G^1]$, we get

\begin{eqnarray}
\label{eqn:strong-f1-g1}
\qquad ax + by + cy + dx \;&>&\; cx + dy + ay + bx \qquad \nonumber \\
\qquad (a+d)(x-y)\;&>&\; (c+b)(x-y) \qquad \qquad \nonumber\\
\qquad a+d\;&>&\; c+b \qquad \qquad\qquad
 \;,
\end{eqnarray}

where the last inequality follows due to the fact that $x>y$. 

Using strong truthfulness condition $E[\trueF, \trueG] > E[F^2, G^2]$, we get

\begin{eqnarray}
\label{eqn:strong-f2-g2}
\qquad by + cy > -dx + a(2y + x) +(g+e-f-h)(-x-y) \qquad
 \,
\end{eqnarray}

and again using strong truthfulness condition $E[\trueF, \trueG] > E[F^3, G^3]$, we get

\begin{eqnarray}
\label{eqn:strong-f3-g3}
\qquad by + cy > -ax + d(2y + x) +(f+h-g-e)(-x-y) \qquad 
\end{eqnarray}

Now, multiplying equation~\ref{eqn:strong-f1-g1} by $y$ and combining equation it with equation~\ref{eqn:strong-f2-g2}, we get

\begin{eqnarray}
\qquad -dx + a(2y + x) +(g+e-f-h)(-x-y) \;\;<\;\; by + cy \;\;<\;\; ay+ dy \qquad \nonumber \\
\qquad  \implies -dx + a(2y + x) +(g+e-f-h)(-x-y) \;\;< \;\; ay+ dy \qquad \nonumber\\
\qquad  \implies a(x+y) \;\;< \;\; d(x+y) + (f+h-g-e)(-x-y)  \qquad
\label{eqn:combine-f1-g1-f2-g2}
\end{eqnarray}

Similarly, equation~\ref{eqn:strong-f1-g1} by $y$ and combining equation it with equation~\ref{eqn:strong-f3-g3}, we get

\begin{eqnarray}
\qquad -ax + d(2y + x) +(f+h-g-e)(-x-y) \;\;<\;\; by + cy \;\;<\;\; ay+ dy \qquad \nonumber \\
\qquad  \implies -ax + d(2y + x) +(f+h-g-e)(-x-y) \;\;< \;\; ay+ dy \qquad \nonumber\\
\qquad \implies  d(x+y) + (f+h-g-e)(-x-y) \;\;< \;\; a(x+y)\qquad
\label{eqn:combine-f1-g1-f3-g3}
\end{eqnarray}

Equation~\ref{eqn:combine-f1-g1-f2-g2} and~\ref{eqn:combine-f1-g1-f3-g3} lead to a contradiction, implying that there does not exist any $a,b,c,d,e,f,g,h,i$ that can satisfy these equations simultaneously.
Therefore, for matrices of the above form there do not exist any strongly truthful scoring matrices.
\end{proof}

Figure~\ref{fig:informed-strong-breakdown} evaluates the sign
structure of the $\Delta$ matrix for the 104 MOOC questions described
earlier. The $\genmech$ mechanism is strongly truthful up to paired
permutations when signals are not clustered, and thus in roughly half
of the worlds.

\subsection{Detail-Free Implementation of the $\genmech$ Mechanism}

So far we have assumed that the $\genmech$ mechanism has access to the sign
structure of $\Delta$. In practice, the signs may be unknown, or
partially known (e.g. the designer may know or assume that the
diagonal of $\Delta$ is positive, but be uncertain about other signs).

The $\genmech$ mechanism can be made detail-free in a
straightforward way by estimating correlation and thus the score
matrix from reports; it remains informed truthful if the
number of tasks is large (even allowing for the new concern that reports
affect the estimation of the distribution and thus the choice of score
matrix.)

\begin{definition}[The $\genmech$ Detail-Free Mechanism (CA-DF)]
As usual, we state the mechanism for two agents for notational simplicity:
\begin{enumerate}
\item Each agent completes $m$ tasks, providing $m$ pairs of reports.
\item Randomly split the tasks into sets $A$ and $B$ of equal size.
\item Let $T^A, T^B$ be the empirical joint distributions of reports on the bonus tasks in $A$ and $B$, with $T^A(i,j)$ the observed frequency of signals $i,j$. Also, let $T^A_M, T^B_M$ be the empirical marginal distribution of reports computed on the penalty tasks in $A$ and $B$, respectively, with $T^A_M(i)$ the observed frequency of signal $i$. Note that we only take one sample per task to ensure the independence of samples.
\item Compute the empirical estimate of the Delta matrix, based on reports rather than signals: $\Gamma^A_{ij} = T^A(i,j)-T^A_M(i)T^A_M(j)$, and similarly for $\Gamma^B$.
\item Define score matrices, \emph{swapping task sets}: $S^A=\Sign(\Gamma^B)$, $S^B=\Sign(\Gamma^A)$. Note that $S^A$ does not depend on the reports on tasks in $A$.
\item Apply the $\genmech$ mechanism separately to tasks in set $A$ and set $B$, using score matrix $S^A$ and $S^B$ for tasks in set $A$ and $B$, respectively.
\end{enumerate}
\end{definition}
\begin{lemma}
  \label{lem:detailfree-1}
  For all strategies $F,G$ and all score matrices $S \in \{0,1\}^{n\times n}$, $E(\trueS,\trueF,\trueG) \geq E(S,F,G)$ in the multi-task mechanism, where $E(S,F,G)$ is the expected score of the mechanism with a fixed score matrix $S$.
\end{lemma}
\begin{proof}
The expected score for arbitrary score matrix and strategies is:
\[E(S,F,G) = \sum_{i=1}^n\sum_{j=1}^n\Delta_{ij} S(F_i,G_j)\]
The expected score for truthful reporting with $\trueS$ is
\begin{align*}
E(\trueS,\trueF,\trueG) = \sum_{i=1}^n\sum_{j=1}^n\Delta_{ij} \Sign(\Delta)_{ij} = \sum_{i,j: \Delta_{ij}>0} \Delta_{ij} \ge \sum_{i=1}^n\sum_{j=1}^n\Delta_{ij} S(F_i, G_j),
\end{align*}
where the inequality follows because $S$ is a 0/1 matrix.
\end{proof}

The lemma gives the main intuition for why $\genmech$-DF is informed truthful for large $m$: even if agents could set the score matrix completely independently of their strategies, the ``truthful'' score matrix $\trueS$ is the one that maximizes payoffs. To get a precise result, the following theorem shows that a score matrix ``close'' to $\trueS$ will be chosen with high enough probability. 
\begin{theorem}[Mechanism $\genmech$-DF is ($\epsilon,\delta$)-informed truthful]
\label{thm:01DF-detail-free}
Let $\epsilon>0$ and $\delta>0$ be parameters. Then there exists 
a number of tasks $m=O(n^{3}\log(1/\delta)/\epsilon^2)$ (for $n$ signals), such that with probability at least $1-\delta$,
there is no strategy profile with expected score more than $\epsilon$ above truthful reporting, and any uninformed strategy has expected score strictly less than truthful. Formally, with probability at least $1-\delta$, 
$E(F,G) \le E(\trueF,\trueG)+\epsilon,$
for all strategy pairs $F,G$; for any uninformed strategy $F_0$ (equivalently $G_0$), $E(F_0,G) < E(\trueF,\trueG)$.
\end{theorem}

\begin{proof}
Let $H^A$ and $H^B$ be the (unobserved) joint signal frequencies, which are a sample from the true joint distribution. Let $M^A$ and $M^B$ be the (unobserved) marginal signal frequencies, which are a sample from the true marginal distribution. Finally, let $\DeltaSup{A}$ and $\DeltaSup{B}$ the corresponding empirical Delta matrices. Fixing strategies $F,G$, $S^A$ is a function of $H^B$ and $M^B$, and independent of $H^A$ and $M^A$. This means that we can write the expected score for tasks in $A$ as 
\begin{align}
E(S^A,F,G) = \sum_{i=1}^n\sum_{j=1}^n\Delta_{ij}S^A(F_i,G_j).
\end{align}
By Lemma~\ref{lem:detailfree-1}, we know that $E(\trueS,\trueF,\trueG) \ge E(S,F,G)$ for all $S,F,G$, and will show that once $m$ is large enough, being truthful gets close to this score with high probability.
We have
\begin{align}
 |E(S_A,\trueF,\trueG) - E(\trueS,\trueF,\trueG)| &=  |E(\Sign(\DeltaSup{B}),\trueF,\trueG) - E(\Sign(\Delta),\trueF,\trueG)| \\
 &= |\sum_{i=1}^n\sum_{j=1}^n\Delta_{ij} (\Sign(\DeltaSup{B})_{ij} - \Sign(\Delta)_{ij})| ~.
\end{align}

Therefore, for some accuracy $\epsilon$ and confidence $\delta$, with $m = O(n^3\log(1/\delta)/\epsilon^2)$, we want

\begin{align}
\label{eq:target}
&|\sum_{i=1}^n\sum_{j=1}^n\Delta_{ij} (\Sign(\DeltaSup{B})_{ij} - \Sign(\Delta)_{ij})| \le \epsilon~.
\end{align}

Observe that

\begin{align}
|\sum_{i,j}\Delta_{ij} (\Sign(\DeltaSup{B})_{ij} - \Sign(\Delta)_{ij})| &\le \sum_{i,j}| \Delta_{ij} (\Sign(\DeltaSup{B})_{ij} - \Sign(\Delta)_{ij})| \\
 &\le \sum_{i,j} | \Delta_{ij} - \DeltaSup{B}_{ij}|~.
\end{align}

Therefore, it is sufficient to learn $\DeltaSup{B}$ such that 

\begin{align}
\label{eq:learnt}
&\sum_{i=1}^n\sum_{j=1}^n| \Delta_{ij} - \DeltaSup{B}_{ij}| \le \epsilon~.
\end{align}

We now use a standard result (see e.g.~\cite{devroye2001combinatorial}, Theorems 2.2 and 3.1), that any distribution over finite domain $\Omega$ is learnable within L1 distance $d$ in $O(|\Omega|/d^2)$ samples with high probability, specifically $1-\delta$ with an additional $\log(1/\delta)$ factor. 

Using this result we can learn the joint signal distribution of the agents using $O(9n^2/\epsilon^2)$ samples with accuracy $\epsilon/3$. We can also learn the marginal distribution of agents' signals using $O(9n^3/\epsilon^2)$ samples from the true marginal distribution with accuracy $\epsilon/3n$. With high probability, after these many samples from each of these distributions, we have

\begin{align}
\sum_{i=1}^n\sum_{j=1}^n| P_{ij} - H^{B}_{ij}| &\le \frac{\epsilon}{3} \label{eq:learnt-joint} \\
\sum_{i=1}^n| P_{i} - M^B_{i}| &\le \frac{\epsilon}{3n} \label{eq:learnt-marginal} ~.
\end{align}

Now,

\begin{align}
\sum_{i,j} | \Delta_{ij} - \DeltaSup{B}_{ij}|  &=  \sum_{i,j} |P_{ij} - H^B_{ij} - (P_i P_j - M^B_i M^B_j)|\\ 
 & \le  \sum_{i,j} |P_{ij} - H^B_{ij}| + \sum_{i,j}|P_i P_j - M^B_i M^B_j|  \quad \text{(Triangle Ineq.)}\\
 & \le \frac{\epsilon}{3} + \sum_{i,j}|P_i P_j - M^B_i \left(P_j \pm \frac{\epsilon}{3n} \right) |   \qquad \text{(Using Eq.~\ref{eq:learnt-joint} \& ~\ref{eq:learnt-marginal} )}\\
 & =  \frac{\epsilon}{3} + \sum_{i,j}|P_i P_j - M^B_i P_j \pm M^B_i \frac{\epsilon}{3n} | \\
 & \le  \frac{\epsilon}{3} + \sum_{i,j}| \left(P_i  - M^B_i \right) P_j| + \sum_{i,j} M^B_i \frac{\epsilon}{3n}   \quad \text{(Triangle Ineq.)} \\
 & = \frac{\epsilon}{3} + \sum_{i,j} P_j |  P_i  - M^B_i  | + \sum_{i,j} M^B_i \frac{\epsilon}{3n}  \\
 & = \frac{\epsilon}{3} + \sum_{i,j} P_j | P_i  - M^B_i  | + \sum_{j} \frac{\epsilon}{3n}  \\
 & \leq \frac{\epsilon}{3} + \sum_{j=1}^n \sum_{i=1}^n | P_i  - M^B_i  | + n \frac{\epsilon}{3n}   \qquad \qquad \qquad \text{(} |P_j|\leq 1 \text{)} \\
  & \leq \frac{\epsilon}{3} + \sum_{j=1}^n \frac{\epsilon}{3n} + \frac{\epsilon}{3}  \qquad \qquad\qquad\qquad\qquad\quad\text{(Using Eq.~\ref{eq:learnt-marginal})}\\
    & = \epsilon ~.
\end{align}

We now conclude
\begin{align}
 |E(S_A,\trueF,\trueG) - E(\trueS,\trueF,\trueG)|  \quad
 &\leq  \quad \sum_{i=1}^n\sum_{j=1}^n | \Delta_{ij} - \DeltaSup{B}_{ij}|  \quad
 \leq \epsilon~,
\end{align}

 which implies $E(S_A,\trueF,\trueG) + \epsilon \ge E(S,F,G)$ for all $S,F,G$.

Finally, note that the expected value of uninformed strategies is 0, because $E(S,F^0,G) = 0$ for any uninformed $F^0$, regardless of score matrix, while $\epsilon$ can always be set small enough ensuring that being truthful has positive expected payoff.
\end{proof}

\subsection{Agent heterogeneity}
\label{sec:extensions}

The $\genmech$ mechanism only uses the signs of the entries of $\Delta$ to
compute scores, not the exact values.
 This means that the results can
handle some variability across agent ``sensing technology,'' as long as
the sign structure of the $\Delta$ matrix is uniform across all
pairwise matchings of peers. In the binary signal case, this reduces
to agents having positive correlation between their signals, giving
exactly the heterogeneity results in \citeN{DasguptaGhosh13}. 
Moreover, the agents themselves do not need to know the detailed
signal model to know how to act; as long as they believe that the
scoring mechanism is using the correct correlation structure, they can
be confident in investing effort and simply report their signals
truthfully.

\subsection{Unintended Signals}
\label{subsec:signal-models}

Finally, we discuss a seemingly pervasive problem in peer prediction:
in practice, tasks may have many distinctive attributes on which agents may base their reports, in addition to the intended signal, and yet all models in the literature assume away the possibility that agents
can choose to acquire such unintended signals.
For example, in online peer assessment where students are
asked to evaluate the quality of student assignments, students
could instead base their assessments on the length of an essay
 or the average number of syllables per word. In an image
categorization system, users could base their reports on the color of
the top-left pixel, or the number of kittens present (!), rather than on
the features they are asked to evaluate. Alternative assessments can
benefit agents in two ways: they may require less effort, and they may
result in higher expected scores via more favorable Delta matrices.\footnote{This issue is related to the perennial problem of spurious correlations in classification and regression.}

We can characterize when this kind of manipulation cannot be
beneficial to agents in the $\genmech$ mechanism.  The idea is that
the amount of correlation coupled with variability across tasks should
be large enough for the intended signal.  Let $\eta$ represent a
particular {\em task evaluation strategy}, which may involve acquiring
different signals from the task than intended.  Let $\DeltaSup{\eta}$
be the corresponding $\Delta$ matrix that would be designed if this
was the signal distribution. This is defined on a domain of signals
that may be distinct from that in the designed mechanism.
In comparison, let $\eta^\ast$ define the task evaluation strategy
intended by the designer (i.e., acquiring signals consistent with the
mechanism's message space), coupled with truthful reporting. The
expected payment from this behavior is 
$\sum_{ij: \DeltaSup{\eta^\ast}_{ij}>0} \DeltaSup{\eta^\ast}_{ij}.$

The maximal expected score for an alternate task evaluation strategy
$\eta$ may require a strategy remapping signal pairs in the signal
space associated with $\eta$ to signal pairs in the intended mechanism
(e.g., if the signal space under $\eta$ is different
than that provided by the mechanism's message space). 
The expected payment is bounded above by
$\sum_{ij: \DeltaSup{\eta}_{ij}>0} \DeltaSup{\eta}_{ij}.$
Therefore, if the expected score for the intended $\eta^\ast$ is
higher than the maximum possible score for other $\eta$, there will be
no reason to deviate.

\section{Conclusion}

We study the design of peer prediction mechanisms that leverage signal
reports on multiple tasks to ensure informed truthfulness, where truthful reporting is the
joint strategy with highest payoff across all joint strategies, and
strictly higher payoff than all uninformed strategies (i.e., those
that do not depend on signals or require effort).
We introduce the $\genmech$ mechanism, which
is informed-truthful
in general multi-signal 
domains. The
mechanism reduces to the~\citeN{DasguptaGhosh13} mechanism
in binary domains, is
 strongly truthful in categorical domains,
and maximally strongly truthful among a broad class of
multi-task mechanisms.
We also present a detail-free version of the mechanism that works without knowledge of the signal distribution 
while retaining $\epsilon$-informed truthfulness.
Interesting directions for future work include: (i) adopting a non-binary
model of effort, and (ii) combining learning with models
of agent heterogeneity.


\begin{thebibliography}{00}


\ifx \showCODEN    \undefined \def \showCODEN     #1{\unskip}     \fi
\ifx \showDOI      \undefined \def \showDOI       #1{{\tt DOI:}\penalty0{#1}\ }
  \fi
\ifx \showISBNx    \undefined \def \showISBNx     #1{\unskip}     \fi
\ifx \showISBNxiii \undefined \def \showISBNxiii  #1{\unskip}     \fi
\ifx \showISSN     \undefined \def \showISSN      #1{\unskip}     \fi
\ifx \showLCCN     \undefined \def \showLCCN      #1{\unskip}     \fi
\ifx \shownote     \undefined \def \shownote      #1{#1}          \fi
\ifx \showarticletitle \undefined \def \showarticletitle #1{#1}   \fi
\ifx \showURL      \undefined \def \showURL       #1{#1}          \fi

\bibitem[\protect\citeauthoryear{Cai, Daskalakis, and Papadimitriou}{Cai
  et~al\mbox{.}}{2015}]%
        {cai-stat-estimate15}
{Yang Cai}, {Constantinos Daskalakis}, {and} {Christos Papadimitriou}. 2015.
\newblock \showarticletitle{Optimum Statistical Estimation with Strategic Data
  Sources}. In {\em Proceedings of The 28th Conference on Learning Theory}.
  280--296.
\newblock


\bibitem[\protect\citeauthoryear{Dasgupta and Ghosh}{Dasgupta and
  Ghosh}{2013}]%
        {DasguptaGhosh13}
{Anirban Dasgupta} {and} {Arpita Ghosh}. 2013.
\newblock \showarticletitle{{Crowdsourced Judgement Elicitation with Endogenous
  Proficiency}}. In {\em WWW13}. 1--17.
\newblock

\bibitem[\protect\citeauthoryear{Devroye and Lugosi}{Devroye and
  Lugosi}{2001}]%
        {devroye2001combinatorial}
{L. Devroye} {and} {G. Lugosi}. 2001.
\newblock {\em Combinatorial Methods in Density Estimation}.
\newblock Springer New York.
\newblock
\showISBNx{9780387951171}
\showLCCN{00058306}

\bibitem[\protect\citeauthoryear{Faltings, Pu, and Tran}{Faltings
  et~al\mbox{.}}{2014}]%
        {faltings-hcomp14}
{Boi Faltings}, {Pearl Pu}, {and} {Bao~Duy Tran}. 2014.
\newblock \showarticletitle{{Incentives to Counter Bias in Human Computation}}.
  In {\em HCOMP 2014}. 59--66.
\newblock


\bibitem[\protect\citeauthoryear{Gao, Mao, Chen, and Adams}{Gao
  et~al\mbox{.}}{2014}]%
        {gao-ec14-trick-or-treat}
{Xi~Alice Gao}, {Andrew Mao}, {Yiling Chen}, {and} {Ryan~P Adams}. 2014.
\newblock \showarticletitle{{Trick or Treat : Putting Peer Prediction to the
  Test}}. In {\em EC'14}.
\newblock
\showISBNx{9781450325653}


\bibitem[\protect\citeauthoryear{Gao, Wright, and Leyton-Brown}{Gao
  et~al\mbox{.}}{2016}]%
        {Gao2016}
{Xi~Alice Gao}, {R.~James Wright}, {and} {Kevin Leyton-Brown}. 2016.
\newblock {Incentivizing Evaluation via Limited Access to Ground Truth : Peer
  Prediction Makes Things Worse}.
\newblock Unpublished, U. British Columbia.   (2016).
\newblock
\showISBNx{9781450339360}


\bibitem[\protect\citeauthoryear{Jain and Parkes}{Jain and Parkes}{2013}]%
        {Jain_acm_trans_econ_compu}
{Shaili Jain} {and} {David~C Parkes}. 2013.
\newblock \showarticletitle{{A Game-Theoretic Analysis of the ESP Game}}.
\newblock {\em ACM Transactions on Economics and Computation\/} {1}, 1 (2013),
  3:1--3:35.
\newblock


\bibitem[\protect\citeauthoryear{Jurca and Faltings}{Jurca and
  Faltings}{2005}]%
        {jurca-faltings2005}
{Radu Jurca} {and} {Boi Faltings}. 2005.
\newblock \showarticletitle{{Enforcing truthful strategies in incentive
  compatible reputation mechanisms}}. In {\em WINE’05}, Vol. 3828 LNCS.
  268--277.
\newblock
\showISBNx{3540309004}
\showISSN{03029743}


\bibitem[\protect\citeauthoryear{Jurca and Faltings}{Jurca and
  Faltings}{2009}]%
        {Jurca2009}
{Radu Jurca} {and} {Boi Faltings}. 2009.
\newblock \showarticletitle{{Mechanisms for making crowds truthful}}.
\newblock {\em Journal of Artificial Intelligence Research\/} {34}, 1 (2009),
  209--253.
\newblock


\bibitem[\protect\citeauthoryear{Jurca and Faltings}{Jurca and
  Faltings}{2011}]%
        {jurca2011}
{Radu Jurca} {and} {Boi Faltings}. 2011.
\newblock \showarticletitle{{Incentives for Answering Hypothetical Questions}}.
  In {\em Workshop on Social Computing and User Generated Content, EC-11}.
\newblock


\bibitem[\protect\citeauthoryear{Kamble, Shah, Marn, Parekh, and
  Ramachandran}{Kamble et~al\mbox{.}}{2015}]%
        {Kamble2015}
{Vijay Kamble}, {Nihar Shah}, {David Marn}, {Abhay Parekh}, {and} {Kannan
  Ramachandran}. 2015.
\newblock \showarticletitle{{Truth Serums for Massively Crowdsourced Evaluation
  Tasks}}.
\newblock  (2015).
\newblock
\showURL{%
\url{http://arxiv.org/abs/1507.07045}}


\bibitem[\protect\citeauthoryear{Kong and Schoenebeck}{Kong and
  Schoenebeck}{2016}]%
        {kong2016}
{Yuqing Kong} {and} {Grant Schoenebeck}. 2016.
\newblock \showarticletitle{{A Framework For Designing Information Elicitation
  Mechanism That Rewards Truth-telling}}.
\newblock  (2016).
\newblock
\showURL{%
\url{http://arxiv.org/abs/1605.01021}}


\bibitem[\protect\citeauthoryear{Kong, Schoenebeck, and Ligett}{Kong
  et~al\mbox{.}}{2016}]%
        {ksl}
{Yuqing Kong}, {Grant Schoenebeck}, {and} {Katrina Ligett}. 2016.
\newblock \showarticletitle{Putting Peer Prediction Under the
  Micro(economic)scope and Making Truth-telling Focal}.
\newblock {\em CoRR\/}  {abs/1603.07319} (2016).
\newblock
\showURL{%
\url{http://arxiv.org/abs/1603.07319}}


\bibitem[\protect\citeauthoryear{Kulkarni, Wei, Le, Chia, Papadopoulos, Cheng,
  Koller, and Klemmer}{Kulkarni et~al\mbox{.}}{2013}]%
        {Kulkarni2013}
{Chinmay Kulkarni}, {Koh~Pang Wei}, {Huy Le}, {Daniel Chia}, {Kathryn
  Papadopoulos}, {Justin Cheng}, {Daphne Koller}, {and} {Scott~R. Klemmer}.
  2013.
\newblock \showarticletitle{{Peer and self assessment in massive online
  classes}}.
\newblock {\em ACM TOCHI\/} {20}, 6 (Dec 2013), 1--31.
\newblock
\showISSN{10730516}


\bibitem[\protect\citeauthoryear{Miller, Resnick, and Zeckhauser}{Miller
  et~al\mbox{.}}{2005}]%
        {MRZ2005}
{Nolan Miller}, {Paul Resnick}, {and} {Richard Zeckhauser}. 2005.
\newblock \showarticletitle{{Eliciting informative feedback: The
  peer-prediction method}}.
\newblock {\em Management Science\/}  {51} (2005), 1359--1373.
\newblock


\bibitem[\protect\citeauthoryear{Piech, Huang, Chen, Do, Ng, and Koller}{Piech
  et~al\mbox{.}}{2013}]%
        {Piech2013}
{Chris Piech}, {Jonathan Huang}, {Zhenghao Chen}, {Chuong Do}, {Andrew Ng},
  {and} {Daphne Koller}. 2013.
\newblock \showarticletitle{{Tuned Models of Peer Assessment in MOOCs}}.
\newblock {\em EDM\/} (2013).
\newblock


\bibitem[\protect\citeauthoryear{Prelec}{Prelec}{2004}]%
        {Prelec2004}
{Drazen Prelec}. 2004.
\newblock \showarticletitle{{A Bayesian Truth Serum For Subjective Data}}.
\newblock {\em Science\/} {306}, 5695 (2004), 462.
\newblock


\bibitem[\protect\citeauthoryear{Radanovic and Faltings}{Radanovic and
  Faltings}{2014}]%
        {radanovic-faltings14}
{Goran Radanovic} {and} {Boi Faltings}. 2014.
\newblock \showarticletitle{{Incentives for Truthful Information Elicitation of
  Continuous Signals}}. In {\em AAAI'14}. 770--776.
\newblock


\bibitem[\protect\citeauthoryear{Radanovic and Faltings}{Radanovic and
  Faltings}{2015a}]%
        {Radanovic-sensing2015}
{Goran Radanovic} {and} {Boi Faltings}. 2015a.
\newblock \showarticletitle{{Incentive Schemes for Participatory Sensing}}. In
  {\em AAMAS 2015}.
\newblock


\bibitem[\protect\citeauthoryear{Radanovic and Faltings}{Radanovic and
  Faltings}{2015b}]%
        {radanovic-subjective-aaai15}
{Goran Radanovic} {and} {Boi Faltings}. 2015b.
\newblock \showarticletitle{{Incentives for Subjective Evaluations with Private
  Beliefs}}.
\newblock {\em AAAI'15\/} (2015), 1014--1020.
\newblock


\bibitem[\protect\citeauthoryear{Radanovic, Faltings, and Jurca}{Radanovic
  et~al\mbox{.}}{2016}]%
        {RFJ2016}
{Goran Radanovic}, {Boi Faltings}, {and} {Radu Jurca}. 2016.
\newblock \showarticletitle{{Incentives for Effort in Crowdsourcing using the
  Peer Truth Serum}}.
\newblock {\em ACM TIST\/} January (2016).
\newblock


\bibitem[\protect\citeauthoryear{Russakovsky, Deng, Su, Krause, Satheesh, Ma,
  Huang, Karpathy, Khosla, Bernstein, Berg, and Fei-Fei}{Russakovsky
  et~al\mbox{.}}{2015}]%
        {imagenet-ILSVRC15}
{Olga Russakovsky}, {Jia Deng}, {Hao Su}, {Jonathan Krause}, {Sanjeev
  Satheesh}, {Sean Ma}, {Zhiheng Huang}, {Andrej Karpathy}, {Aditya Khosla},
  {Michael Bernstein}, {Alexander~C. Berg}, {and} {Li Fei-Fei}. 2015.
\newblock \showarticletitle{{ImageNet Large Scale Visual Recognition
  Challenge}}.
\newblock {\em International Journal of Computer Vision (IJCV)\/} (April 2015),
  1--42.
\newblock


\bibitem[\protect\citeauthoryear{Shnayder, Agarwal, Frongillo, and
  Parkes}{Shnayder et~al\mbox{.}}{2016a}]%
        {shnayder-working16}
{Victor Shnayder}, {Arpit Agarwal}, {Rafael Frongillo}, {and} {David~C.
  Parkes}. 2016a.
\newblock \showarticletitle{Informed Truthfulness in Multi-Task Peer
  Prediction}.
\newblock  (2016).
\newblock
\showURL{%
\url{https://arxiv.org/abs/1603.03151}}


\bibitem[\protect\citeauthoryear{Shnayder, Frongillo, and Parkes}{Shnayder
  et~al\mbox{.}}{2016b}]%
        {shnayder-ijcai16}
{Victor Shnayder}, {Rafael Frongillo}, {and} {David~C. Parkes}. 2016b.
\newblock \showarticletitle{Measuring Performance Of Peer Prediction Mechanisms
  Using Replicator Dynamics}.
\newblock {\em IJCAI-16\/} (2016).
\newblock


\bibitem[\protect\citeauthoryear{von Ahn and Dabbish}{von Ahn and
  Dabbish}{2004}]%
        {vonAhn2004}
{Luis von Ahn} {and} {Laura Dabbish}. 2004.
\newblock \showarticletitle{Labeling Images with a Computer Game}. In {\em
  CHI'04}. ACM, New York, NY, USA, 319--326.
\newblock
\showISBNx{1-58113-702-8}


\bibitem[\protect\citeauthoryear{Waggoner and Chen}{Waggoner and Chen}{2014}]%
        {waggoner14}
{Bo Waggoner} {and} {Yiling Chen}. 2014.
\newblock \showarticletitle{{Output Agreement Mechanisms and Common
  Knowledge}}. In {\em HCOMP'14}.
\newblock


\bibitem[\protect\citeauthoryear{Witkowski and Parkes}{Witkowski and
  Parkes}{2012}]%
        {RBTS-Witkowski2012}
{Jens Witkowski} {and} {David~C Parkes}. 2012.
\newblock \showarticletitle{{A Robust Bayesian Truth Serum for Small
  Populations}}. In {\em AAAI'12.}
\newblock


\bibitem[\protect\citeauthoryear{Witkowski and Parkes}{Witkowski and
  Parkes}{2013}]%
        {Witkowski2013}
{Jens Witkowski} {and} {David~C Parkes}. 2013.
\newblock \showarticletitle{{Learning the Prior in Minimal Peer Prediction}}.
  In {\em EC'13}.
\newblock


\bibitem[\protect\citeauthoryear{Wright, Thornton, and Leyton-Brown}{Wright
  et~al\mbox{.}}{2015}]%
        {Wright-KLB2015}
{James~R Wright}, {Chris Thornton}, {and} {Kevin Leyton-Brown}. 2015.
\newblock \showarticletitle{{Mechanical TA : Partially Automated High-Stakes
  Peer Grading}}. In {\em SIGSCE'15}.
\newblock
\showISBNx{9781450329668}


\bibitem[\protect\citeauthoryear{Wu, Tzamos, Daskalakis, Weinberg, and
  Kaashoek}{Wu et~al\mbox{.}}{2015}]%
        {wu-las2015}
{William Wu}, {Christos Tzamos}, {Constantinos Daskalakis}, {Matthew Weinberg},
  {and} {Nicolaas Kaashoek}. 2015.
\newblock \showarticletitle{{Game Theory Based Peer Grading Mechanisms For
  MOOCs}}. In {\em Learning@Scale 2015}.
\newblock
\showISBNx{9781450334112}


\end{thebibliography}
\end{document}